\documentclass[preprint,12pt]{elsarticle}
\usepackage{amsthm}
\usepackage{amsmath}
\usepackage{amsopn}   

\newcommand{\E}{\mathbb{E}}

\newcommand{\dd}{\mathrm{d}}

\DeclareMathOperator{\argmin}{arg\,min}

\theoremstyle{plain}

\newtheorem{proposition}{Proposition}
\newtheorem{lemma}{Lemma}
\theoremstyle{remark}

\usepackage{amssymb}
\usepackage{amsmath}
\usepackage{url}
\usepackage{xcolor}
\usepackage{bbm}

\usepackage{graphicx}
\usepackage{subcaption}

\journal{\textcolor{blue}{Energy Conversion and Management}}

\begin{document}

\begin{frontmatter}



\title{Weibull-Stationary Stochastic Differential Equations for Conditional Long-Horizon Wind Power Forecasting}


\author{
Luca Di Persio\footnotemark[1]\,   Email: luca.dipersio@univr.it \\
Mehrdad Ghadiri\footnotemark[1]\,  Email: mehrdad.ghadiri@studenti.univr.it}

\affiliation{organization={Department of Computer Science, College of Mathematics, University of Verona},
            addressline={Strada le Grazie 15}, 
            city={Verona},
            postcode={37134}, 
            state={Verona},
            country={Italy}}
 

\begin{abstract}

We present a one-month-ahead conditional probabilistic framework for wind-power forecasting at ten-minute resolution. Monthly Weibull shape and scale parameters are estimated from serially dependent SCADA wind-speed data, corrected through a Godambe covariance, and forecast by a heteroskedastic Kalman filter on a bivariate VAR(1) state-space model. Conditional on the MMSE forecasted Weibull invariant law, we construct and compare three positive wind-speed SDE models: an Ornstein-Uhlenbeck-Weibull transform, a Fokker-Planck drift-first diffusion, and a Fokker-Planck diffusion-first model. The simulated wind-speed ensembles are mapped to power through a calibrated XGBoost power curve. Applied to January 2021 data from a Senvion MM92 turbine at Kelmarsh Wind Farm, the three SDE formulations are statistically indistinguishable in probabilistic accuracy, with mean CRPS values between 1.569 and 1.575 m/s. The diffusion-first model is therefore preferred on computational grounds, reducing runtime by about a factor of seven relative to the OU-Weibull model. In the power domain, the Wasserstein distance between simulated and observed distributions is 26.1-27.6 kW, below $1.4\%$ of rated capacity, while the monthly energy-yield bias is about $-7.3\%$ for the examined month. Exceedance-probability errors remain below 1.6 percentage points over the 0-1500 kW range and about 2.2 percentage points near rated power. These quantities provide decision-relevant probabilistic inputs for downstream operational problems, rather than completed reserve, storage, market, or fatigue-optimization decisions. Full marginalisation over the Kalman predictive law of the Weibull parameters is left as a natural extension.

\end{abstract}












\end{frontmatter}


\section{Introduction}
\label{sec:Introduction}

Wind energy has become one of the most important pillars of the global transition to clean and sustainable power. Over the past decade, we have witnessed remarkable growth in installed wind capacity, driven by continuous improvements in turbine technology, steadily declining costs, and the urgent societal need to reduce carbon emissions. A recent study published in \textit{Nature Energy} confirms that global wind power has 
expanded at a pace that now aligns with the growth trajectories required for major climate-mitigation goals, which underscores just how central wind energy has become to building a reliable and low-carbon energy future~\cite{JakhmolaWind2026}. As wind penetration continues to rise, accurate wind-power forecasting is no longer a luxury but a practical necessity. Grid operators depend on reliable forecasts to anticipate output fluctuations, maintain system stability, schedule reserves, and plan energy storage~\cite{SelvarajWind2026}. Without this foresight, the variability inherent in wind generation becomes a direct operational and financial burden.

What makes wind-power forecasting genuinely difficult is not simply the randomness of wind, but the fact that wind power exhibits highly nonlinear dynamics that make deterministic 
prediction increasingly unreliable as the forecast horizon grows~\cite{AI2023114222, ZHU2024118062}. This challenge is well recognised in the literature: as the forecasting horizon extends beyond a few hours, prediction errors accumulate at every step, and any deterministic model, regardless of its structural sophistication, begins to lose its grip on the true trajectory of the process~\cite{XIE2023105804, SAEED2025137979, DantasBrowell2026}. For short-term horizons of minutes to a few hours, statistical and machine learning models have demonstrated strong performance and remain the workhorses of operational forecasting~\cite{ARSLANTUNCAR2024197, GUO20253753}. More recently, hybrid approaches that combine stochastic differential equations with deep learning architectures have been proposed to improve short-term accuracy further by explicitly representing the stochastic nature of wind dynamics~\cite{electricity6030048}. However, once we move beyond short horizons, the picture changes fundamentally, and this is precisely the challenge that motivates our work.

For forecasting horizons of one month or more, we argue that the appropriate shift is away from deterministic multi-step trajectory prediction and toward probabilistic frameworks that characterise the distribution of possible future outcomes, rather than committing to a single realisation~\cite{DiPersio2024WindItaly,Ceresa2024}. This view is not merely our preference; it is supported by a growing body of evidence. Iversen et al.\ demonstrate rigorously that probabilistic models outperform deterministic machine learning methods in multi-step 
forecasting because they are able to capture inter-temporal uncertainty and error propagation across extended 
horizons~\cite{Iversen2017}. Similarly, Wu et al.\ emphasise that wind farm power generation is highly volatile and stochastic due to the influence of meteorological factors, which makes it essential to incorporate weather-driven uncertainty explicitly into any long-term forecasting model~\cite{11088843}. Also Ladopoulou et al. further demonstrate that the wind speed-to-power relationship itself exhibits input-dependent noise and non-stationary correlation structure across operating regimes, reinforcing the broader case that probabilistic formulations are essential at every stage of the wind power modelling chain \cite{ladopoulou2025probabilisticwindpowermodelling}. We find these arguments compelling, and they directly motivate the probabilistic architecture we develop in this paper.

Within the probabilistic modelling tradition, stochastic differential equations have emerged as a particularly attractive tool for wind-speed simulation because they provide a principled way to combine physically motivated dynamics with a prescribed stationary distribution. The foundational contribution in this direction was made by Z\'{a}rate-Mi\~{n}ano et al., who introduced a continuous-time stochastic wind-speed model based on an Ornstein-Uhlenbeck process combined with a memoryless transformation~\cite{ZARATEMINANO201342}. This elegant construction allows the generation of autocorrelated wind-speed trajectories that follow a Weibull stationary distribution, making the model well-suited for power system dynamic studies. Building on this foundation, the same authors later extended their framework by exploiting the Regression Theorem and the stationary Fokker-Planck equation to derive closed-form expressions for both drift and diffusion terms under a broad class of target distributions~\cite{ZARATEMINANO2016186}. Arenas-L\'{o}pez and Badaoui subsequently advanced this line of work by evaluating multiple probability density functions within both the Ornstein-Uhlenbeck memoryless-transform and Fokker-Planck frameworks. They also emphasized that accurate stochastic simulation of the wind-speed process is the central modelling task, since power output is then obtained as a deterministic consequence of the wind-to-power map~\cite{ARENASLOPEZ2020113152, ARENASLOPEZ2020118842}.

Despite this progress, we believe that a fundamental limitation 
persists throughout this line of research, and it has not been adequately addressed in the existing literature. In 
essentially all prior work, the parameters of the wind speed 
probability distribution are estimated from historical observations and then held fixed, meaning that these models reflect past wind conditions rather than anticipating future variability. Furthermore, the final output is typically reduced to a single point estimate of energy production, rather than being embedded within a formal probabilistic forecasting framework that communicates uncertainty in a decision-relevant way. The distributional parameters are treated as constants over a given historical period, which means that important temporal dynamics, including seasonal fluctuations, inter-monthly variability, and the propagation of estimation 
uncertainty through the simulation chain, are ignored. 
We find this particularly problematic for operational applications such as reserve scheduling, storage sizing, and maintenance planning, where it is precisely the uncertainty about future conditions, not the characterisation of past ones, that drives decision making.

This paper addresses that gap through a conditional probabilistic pipeline whose numerical core is deliberately conditioned on the MMSE forecasted Weibull invariant law. The first stage estimates monthly Weibull shape and scale parameters from SCADA wind-speed observations and corrects the covariance of the monthly MLEs for serial dependence. The second stage models the joint log-parameter dynamics through a heteroskedastic Kalman state-space VAR(1), thereby producing a predictive law for the next-month Weibull parameters and, in particular, its MMSE value $(\hat{k}_{49},\hat{\lambda}_{49})$. The third stage conditions on this forecasted invariant law and constructs three positive wind-speed processes with Weibull stationary distribution: an Ornstein-Uhlenbeck-Weibull transport process, a Fokker-Planck drift-first diffusion, and a Fokker-Planck diffusion-first diffusion. The fourth stage pushes the simulated wind-speed ensemble through a calibrated deterministic XGBoost power curve. Thus, parameter uncertainty is quantified at the Kalman layer, while the numerical wind-speed and power forecasts reported in this paper propagate the conditional SDE path uncertainty under the selected invariant law. Full marginalisation over the Kalman predictive distribution of $(k,\lambda)$ is mathematically natural, but is not part of the present numerical study.

The remainder of this paper is organised as follows. Section~2 describes the Kelmarsh dataset and turbine specifications. Section~3 presents the Weibull parameter estimation and state-space forecasting stage. Section~4 derives and calibrates the three SDE formulations and evaluates their wind-speed simulation performance. Section~5 presents the probabilistic 
wind-power results and their operational implications. 
Section~6 concludes and identifies directions for future research.

Before proceeding, it is useful to make explicit the probabilistic composition underlying the paper. The historical SCADA record is first mapped into a Kalman predictive law for the Weibull parameter vector. In the numerical study, this law is summarized by its MMSE value, which defines the invariant Weibull distribution imposed on the SDE simulations. The corresponding conditional wind-speed ensemble is then pushed forward through the fitted deterministic power curve. The latter is the precise sense in which the reported power forecast is probabilistic: it is a conditional probabilistic forecast given the forecasted invariant Weibull law and the fitted power map.

\section{Case Study, Dataset and Preprocessing}

To analyse and validate the framework, we use SCADA data from the Kelmarsh Wind Farm in Northamptonshire, United Kingdom, publicly available through the Zenodo repository~\cite{zenodo_dataset}. The dataset covers six 2.05 MW Senvion MM92 turbines from January 2017 through June 2021, and it contains ten-minute measurements of hub-height wind speed, active power, rotor speed, blade pitch angle, nacelle orientation, and related operational variables. We focus on Turbine 1 as a representative and conveniently documented test case. This choice is not meant to establish turbine-level universality; the conclusions below should therefore be read as a proof of concept for this turbine and month, with cross-turbine validation left for subsequent work.

Since the raw SCADA data inevitably contains abnormal 
observations arising from sensor faults, communication dropouts, maintenance events, and curtailment periods~\cite{MORRISON2022473}, we apply a systematic four-step cleaning pipeline before any statistical modelling is performed. In the first step, we carry out physical range filtering, discarding any observation that violates the known 
operating envelope of the Senvion MM92, such as wind speeds outside the admissible range, negative power values, or power readings exceeding rated capacity~\cite{Gill2012}. In the second step, we perform abnormal pattern detection using the 
color-space image-based algorithm proposed by Long 
et al.~\cite{LONG2022118594}, which constructs a 
three-dimensional wind power curve image from wind speed, power output, and data frequency, and then applies Canny edge detection and Hough transform to identify and isolate stacked outliers that would otherwise remain hidden within the physical range. In the third step, we remove curtailment and maintenance periods by flagging intervals in which the power output is zero while the wind speed exceeds the cut-in threshold of 3.0~m/s. As Morrison et al.~\cite{MORRISON2022473} discuss, these observations are better interpreted as intentional shutdowns or curtailment events rather than genuine calm periods, and they must therefore be treated as contextual anomalies 
before any power curve analysis can proceed. In the fourth and final step, we substitute the remaining missing power values using an XGBoost regression model trained exclusively on the clean, retained observations from the previous three steps, using the concurrent wind speed recorded by the SCADA system as the primary input. Since wind power is a deterministic function of wind speed through 
the turbine power curve~\cite{Burton2011}, the model is able to recover a physically consistent power estimate for each missing interval without requiring any information from the forecast period. 
It is important to note that the imputed power values are used solely to preserve continuity in descriptive plots and time-series visualisations. The final power-curve XGBoost model and all reported validation metrics, including RMSE, MAE, and $R^2$, are computed exclusively on the observed clean power values retained after the first three cleaning steps.

\begin{figure}[htbp]
    \centering
        \includegraphics[width=9 cm, height=6.5 cm]
        {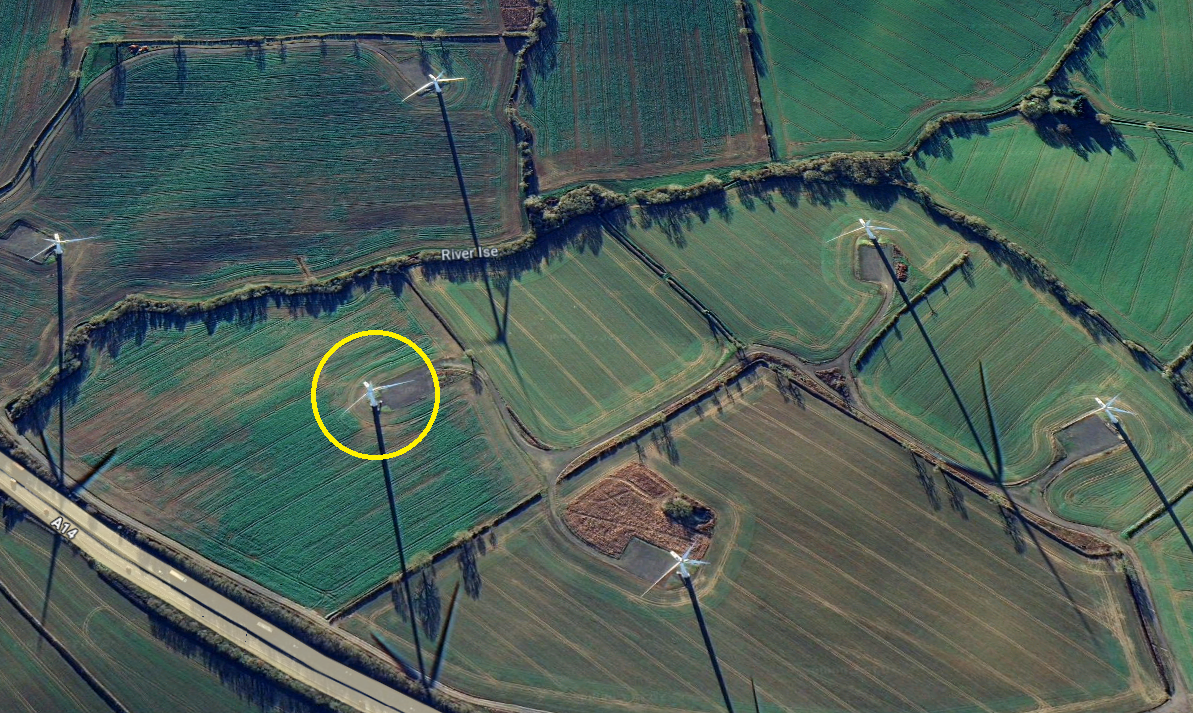}
        
        \caption{\small Zoomed-in satellite view showing the Kelmarsh Wind Farm\cite{KelmarshWindFarm}. }
        \label{fig:KelmarshWindFarm}

\end{figure}

\begin{table}[ht]
\centering
\caption{Wind Turbine specifications. \cite{windturbinemodels_senvionmm92}}
\label{tab:turbine_specs}

\resizebox{\textwidth}{!}{
\begin{tabular}{lcccccc}
\hline
Parameter & Rated Power (kW) & Cut-in Speed (m/s) & Rated Speed (m/s) & Rotor Diameter (m) & Swept Area (m$^2$) & Blades \\
\hline
Value     & 2,050.0 & 3.0 & 12.5 & 92.5 & 6,720.0 & 3 \\
\hline
\end{tabular}
}

\end{table}

\section{Forecasting of Weibull Parameters $(k_m, \lambda_m)$}
\label{sec:weibull}

\paragraph{Weibull Distribution}

The Weibull distribution is the established probabilistic model for wind speed, validated across a wide range of geographical and meteorological conditions \cite{seguro2000modern, carta2009review}. Its two parameters carry direct physical interpretations: the scale parameter governs mean wind resource intensity, while the shape parameter encodes variability around that mean \cite{ALJEDDANI202356}.  Capturing their joint temporal evolution is therefore essential for monthly forecasting.

\paragraph{Parameter estimation of Weibull Distribution}

Let $\mathcal{T}_m$ denote the set of ten-minute wind-speed observations recorded in month $m$. Within each monthly period, wind speed is modelled as a Weibull random
variable with shape parameter $k_m > 0$ and scale parameter $\lambda_m > 0$, whose
density is
\begin{equation}
    f(v;\,k,\lambda)
    = \frac{k}{\lambda}
      \left(\frac{v}{\lambda}\right)^{k-1}
      \exp\!\left[-\left(\frac{v}{\lambda}\right)^k\right],
    \qquad v \ge 0.
    \label{eq:weibull-density}
\end{equation}
The parameters are estimated by maximum likelihood, maximising the log-likelihood
\begin{equation}
    \log L(k,\lambda)
    = \sum_{t \in \mathcal{T}_m}
      \left[
        \log k - \log\lambda
        + (k-1)\log\!\left(\frac{v_t}{\lambda}\right)
        - \left(\frac{v_t}{\lambda}\right)^k
      \right].
    \label{eq:mle}
\end{equation}
Wind-speed observations within a month are serially correlated rather than independent \cite{RAMIREZ20052419}, so the classical MLE covariance formula, which assumes independent and identically distributed observations, underestimates the true estimation uncertainty. Under standard weak-dependence conditions (stationarity and strong mixing), the MLE remains consistent and asymptotically normal, but its covariance is described by the sandwich, or Godambe, matrix rather than the inverse Fisher information alone \cite{White1982}. Concretely, if $\theta_m = (k_m, \lambda_m)^\top$ denotes the true parameter vector,
the MLE $\widehat{\theta}_m$ satisfies
\begin{equation}
   \sqrt{|\mathcal{T}_m|}\,\bigl(\widehat{\theta}_m - \theta_m^\star\bigr)
   \;\Rightarrow\;
   \mathcal{N}\!\Bigl(0,\;
   I(\theta_m^\star)^{-1}\,J(\theta_m^\star)\,I(\theta_m^\star)^{-1}
   \Bigr),
   \label{eq:sandwich}
\end{equation}

where
\begin{equation}
I(\theta) = -\mathbb{E}\left[\nabla^2 \log f(V_t;\theta)\right]
\end{equation}

\begin{equation}
J(\theta) = \sum_{\ell \in \mathbb{Z}} 
\mathrm{Cov}\!\left(S_t(\theta), S_{t+\ell}(\theta)\right),
\qquad
S_t(\theta)=\nabla_\theta\log f(V_t;\theta).
\end{equation}

Here, $I(\theta)$ is the Fisher information matrix and $J(\theta)$ is the long-run covariance of the score process, accounting for serial dependence in the observations. The additional factor $J(\theta)$ relative to the standard formula inflates the variance of the MLE, reflecting the reduced effective sample size caused by temporal autocorrelation. In practice, $J(\theta)$ is estimated by the Newey-West HAC estimator with a Bartlett kernel and 
bandwidth $L = \lfloor 1.1447(\hat{\rho}^2|T_m|)^{1/3}\rfloor$, where 
$\hat{\rho}$ is the first-order autocorrelation of the score; autocovariances are normalised by $|T_m|$ throughout to avoid variance inflation at large lags 
\cite{Andrews1991}.

\paragraph{Log-space transformation}

All subsequent modelling is performed on the log-transformed parameter vector $\eta_m = (\log k_m,\, \log \lambda_m)^\top$ rather than on the original scale. Working in log-space serves two purposes: it maps the strictly positive parameters onto the
entire real line, removing the positivity constraint from the time-series model, and it stabilises the variance of the estimates across months with different sample sizes.
Propagating the MLE covariance to the log scale via the delta method gives
\begin{equation}
    \mathrm{Var}(\widehat{\eta}_m)
    \;\approx\;
    |\mathcal{T}_m|^{-1}\,
    D(\widehat{\theta}_m)\,
    \widehat{I}(\widehat{\theta}_m)^{-1}\,
    \widehat{J}(\widehat{\theta}_m)\,
    \widehat{I}(\widehat{\theta}_m)^{-1}\,
    D(\widehat{\theta}_m)^\top
    \;\equiv\;
    |\mathcal{T}_m|^{-1}\,\widehat{V}_m,
    \label{eq:delta-method}
\end{equation}
where $D(\widehat{\theta}_m) = \mathrm{diag}(1/\widehat{k}_m,\, 1/\widehat{\lambda}_m)$
is the Jacobian of the log transformation evaluated at the MLE. The month-specific matrix $|\mathcal{T}_m|^{-1}\hat{V}_m$ — not $\hat{V}_m$ itself — is used as the observation noise covariance in the state-space model below, making the downstream filtering step fully aware of the varying estimation precision across months.

\paragraph{State-space model and Kalman filtering}
The forty-eight monthly log-parameter vectors $\widehat{\eta}_1, \dots, \widehat{\eta}_{48}$ are treated as noisy observations of a latent bivariate process that evolves according to a vector autoregression of order $p$. A VAR model is preferred over two independent univariate models for two complementary reasons. First, $k_m$ and $\lambda_m$ are empirically correlated across months: a month with unusually high mean wind speed typically exhibits a different variability structure as well, so modelling them jointly captures the cross-parameter dependence and produces forecasts that are internally consistent. Second, with only 48 observations, the added complexity of a VARMA specification — which requires the simultaneous identification and estimation of both autoregressive and moving-average coefficient matrices — leads to over-parameterization and numerical instability in small samples; the simpler VAR structure avoids these difficulties while retaining sufficient flexibility to characterise the joint dynamics of the two parameters \cite{Lutkepohl2005}.

To enable Kalman filtering, the VAR($p$) model is cast in companion form by stacking the current and $p-1$ lagged log-parameter vectors into an augmented state vector
$x_m \in \mathbb{R}^{2p}$. The system then takes the standard linear Gaussian
state-space form, with a state transition equation
\begin{equation}
    x_m = c + F\,x_{m-1} + u_m,
    \qquad u_m \sim \mathcal{N}(0,\,Q),
    \label{eq:state}
\end{equation}
where the companion matrix $F$ encodes the autoregressive coefficients at all lags, the intercept vector $c$ absorbs the unconditional mean, and the process noise covariance $Q$ is nonzero only in the block corresponding to the current innovation, and an observation
equation
\begin{equation}
    \widehat{\eta}_m = H\,x_m + e_m,
    \qquad
    e_m \sim \mathcal{N}\!\left(0,\; |\mathcal{T}_m|^{-1}\widehat{V}_m\right),
    \label{eq:obs}
\end{equation} 
where $H = (I_2,\,0,\ldots,0)$ selects the current log-parameter pair from the state
vector. The observation noise covariance equation \eqref{eq:obs} is heteroskedastic: it varies month 
by month according to the scaled delta-method sandwich estimate 
$|\mathcal{T}_m|^{-1}\hat{V}_m$, which
means that months with many observations and low serial correlation receive greater weight
in the filter than months with few observations or strong dependence. This is the reason why the Kalman filter is preferred here over a direct least-squares fit to the MLE series: ordinary least squares treats all monthly estimates as equally reliable, whereas the Kalman filter weights each observation by its actual precision.

Given the fitted model parameters, the Kalman filter recursions propagate the posterior
mean and covariance of the latent state forward through the forty-eight months of
training data. The one-step-ahead forecast of the log-parameter vector for the target
month is
\begin{equation}
    \widehat{x}_{m+1|m} = c + F\,\widehat{x}_{m|m},
    \label{eq:forecast}
\end{equation}
and the corresponding forecast of the Weibull parameters is recovered by
back-transformation,
\begin{equation}
    \bigl(\widehat{k}_{m+1},\;\widehat{\lambda}_{m+1}\bigr)
    =
    \Bigl(
    \exp\!\bigl(\widehat{\eta}_{1,\,m+1|m}\bigr),\;
    \exp\!\bigl(\widehat{\eta}_{2,\,m+1|m}\bigr)
    \Bigr).
    \label{eq:backtransform}
\end{equation}
The exponential back-transformation in equation~\eqref{eq:backtransform} guarantees that the forecasted shape and scale parameters are strictly positive regardless of the value of the latent state, so that equation~\eqref{eq:weibull-density} always defines a proper probability density. The forecasted pair $(\widehat{k}_{m+1}, \widehat{\lambda}_{m+1})$ is then passed to the stochastic differential equation simulation stage, where it prescribes the target invariant distribution that each simulated wind-speed trajectory must respect.

\subsection{Analysis of Estimated  Weibull Parameters}

Before analysing the estimation of the Weibull parameters and the corresponding probability distribution for January 2021, it is useful to first examine the January wind-speed distributions from 2017-2020, shown in Figure~\ref{fig:weibullDistribution}. Although the fitted Weibull parameters vary from year to year, the model reproduces the observed wind-speed distributions well. These variations reveal differences in both the characteristic wind strength, described by the scale parameter $\lambda$, and the regularity or variability of wind speeds, described by the shape parameter $k$, across the January months.

\begin{figure}[h!]
\centering
\includegraphics[width=1.0\linewidth]{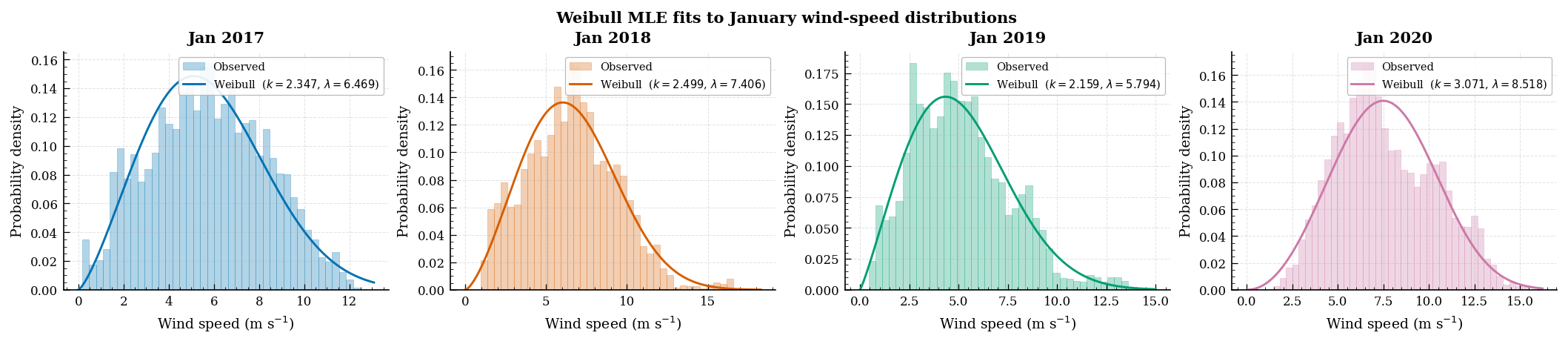}
\caption{Weibull Distribution across January months (2017-2020).}
\label{fig:weibullDistribution}
\end{figure}

To estimate the Weibull parameters for January 2021, we adopt a bivariate VAR(1) model for the joint log-parameter dynamics. This choice is motivated by two considerations: first, it is the minimal specification that captures cross-parameter dependence between 
$\log \hat{k}_m$ and $\log \hat{\lambda}_m$; second, with only $M = 48$ monthly observations, higher-order VAR or VARMA specifications introduce more free parameters than the sample can reliably support \cite{Lutkepohl2005}.

\paragraph{State-Space VAR(1) Forecast and Kalman Filter Output}

\begin{figure}[h!]
\centering
\includegraphics[width=0.75\linewidth]{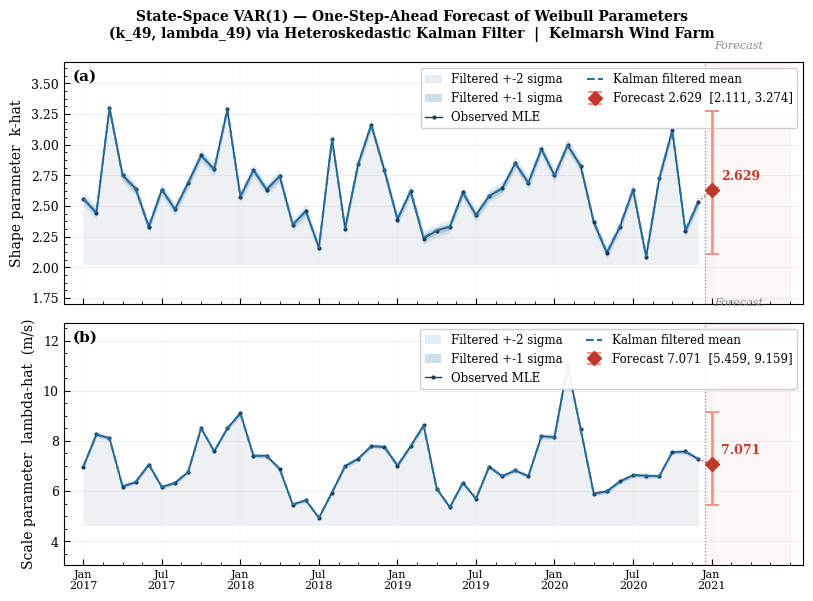}
\caption{State-space VAR(1) one-step-ahead Kalman filter forecast of Weibull shape $\hat{k}_{49}$ (a) and scale $\hat{\lambda}_{49}$ (b) parameters for January 2021.} 
\label{fig:forecast}
\end{figure}

Figure~\ref{fig:forecast} presents the one-step-ahead forecast of the Weibull parameters for January 2021, produced by the heteroskedastic Kalman filter described in Section \ref{sec:weibull}. The Kalman filtered mean tracks
the observed MLE series closely throughout the training period, with the posterior standard deviation remaining consistently small — approximately 0.01 in log-space — relative to the natural month-to-month variability of the series. This behaviour is a consequence of the high observation
precision associated with the large monthly sample sizes (approximately 4,000 ten-minute intervals per month), which renders the MLE estimates highly reliable \cite{lehmann1998theory} and causes the filter to assign dominant weight to the observations at each update step \cite{Kalman1960}.

The one-step-ahead Kalman prediction for January~2021, obtained by propagating the filtered state at month~48 forward through the VAR(1) transition equation~\eqref{eq:state} without access to any January~2021 observations, is recovered via the exponential back-transformation equation~\eqref{eq:backtransform}, 
with $95\%$ prediction intervals of $[2.111,\,3.274]$ and 
$[5.459,\,9.159]$~m\,s$^{-1}$ for the shape and scale parameters respectively. The width of these intervals reflects genuine uncertainty about the January wind 
regime arising from the VAR(1) process noise covariance $Q$, which governs month-to-month parameter dynamics, rather than from observation noise, which the posterior has effectively collapsed to approximately $0.01$ in log-space. The scale parameter interval spans approximately $3.7$~m\,s$^{-1}$, corresponding to a relative uncertainty of roughly $52\%$ of the forecast value, 
consistent with the interannual variability visible in the historical record~\citep{CARTA2009933}.

\paragraph{Kalman Filter Uncertainty Structure}

\begin{figure}[h!]
\centering
\includegraphics[width=0.75\linewidth]{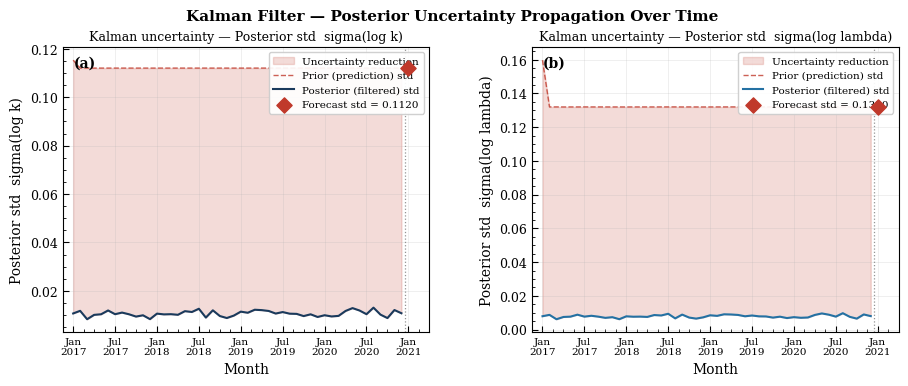}
\caption{Prior and posterior standard deviations of $\log\hat{k}$ (a) and $\log\hat{\lambda}$ (b) throughout the Kalman filter recursion.} 
\label{fig:uncertainty}
\end{figure} 

Figure~\ref{fig:uncertainty} illustrates the evolution of the prior and posterior standard deviations of the log-space parameter estimates throughout the filtering recursion. A striking feature of this figure is the large gap between the prior standard deviation, approximately 0.11 for log\,$\hat{k}$ and 0.13 for log\,$\hat{\lambda}$, and the posterior standard deviation of approximately 0.01 in both cases. This gap represents the uncertainty removed by each monthly observation, and its magnitude reflects the high precision of the monthly MLE estimates. The
forecast standard deviations at month 49, which equal 0.112 and 0.130 for log\,$\hat{k}$ and log\,$\hat{\lambda}$ respectively, are determined almost entirely by the VAR(1) process noise covariance $Q$ rather than by residual observation uncertainty, since the posterior has effectively collapsed the latter. This behaviour validates the theoretical design of the pipeline: the predictive covariance $\mathbf{P}_{49|48}$ quantifies the month-to-month uncertainty in the forecasted Weibull parameters, while the numerical SDE experiments below condition on the MMSE value of this predictive law.
\paragraph{Selection of the Point Forecast as the Target Invariant Law}

In the numerical SDE study, the wind-speed simulation is conditioned on the single parameter pair $(\hat{k}_{49},\hat{\lambda}_{49})=(2.6272,7.0691\,\mathrm{m\,s}^{-1})$ rather than on the full Kalman predictive distribution $\Pi_{49}=\mathcal N(\hat{\boldsymbol\eta}_{49|48},\mathbf P_{49|48})$, which is a plug-in decision. If the parameter vector used in the simulation is evaluated under posterior squared-error loss, the optimal representative on the log scale is
\begin{equation}
\boldsymbol\eta^*=\argmin_{\boldsymbol\eta}\,\mathbb E\!\left[\|\boldsymbol\eta-\boldsymbol\eta_{49}\|^2\mid \mathcal F_{48}\right]
=\mathbb E[\boldsymbol\eta_{49}\mid\mathcal F_{48}]
=\hat{\boldsymbol\eta}_{49|48}.
\end{equation}
For a Gaussian predictive law, this representative also coincides with the posterior mode and componentwise median. This argument justifies the chosen invariant law as a point decision, but it does not turn the subsequent SDE ensemble into a fully marginal forecast over parameter uncertainty. The ensemble generated below should therefore be interpreted as the conditional law of the wind-speed process given the MMSE forecasted Weibull invariant distribution. The full predictive mixture would draw $(k,\lambda)$ from $\Pi_{49}$ before each SDE simulation, and is left as a natural extension because it would require a separate numerical study.

\paragraph{Predictive Weibull Distribution}

\begin{figure}[h!]
\centering
\includegraphics[width=0.75\linewidth]{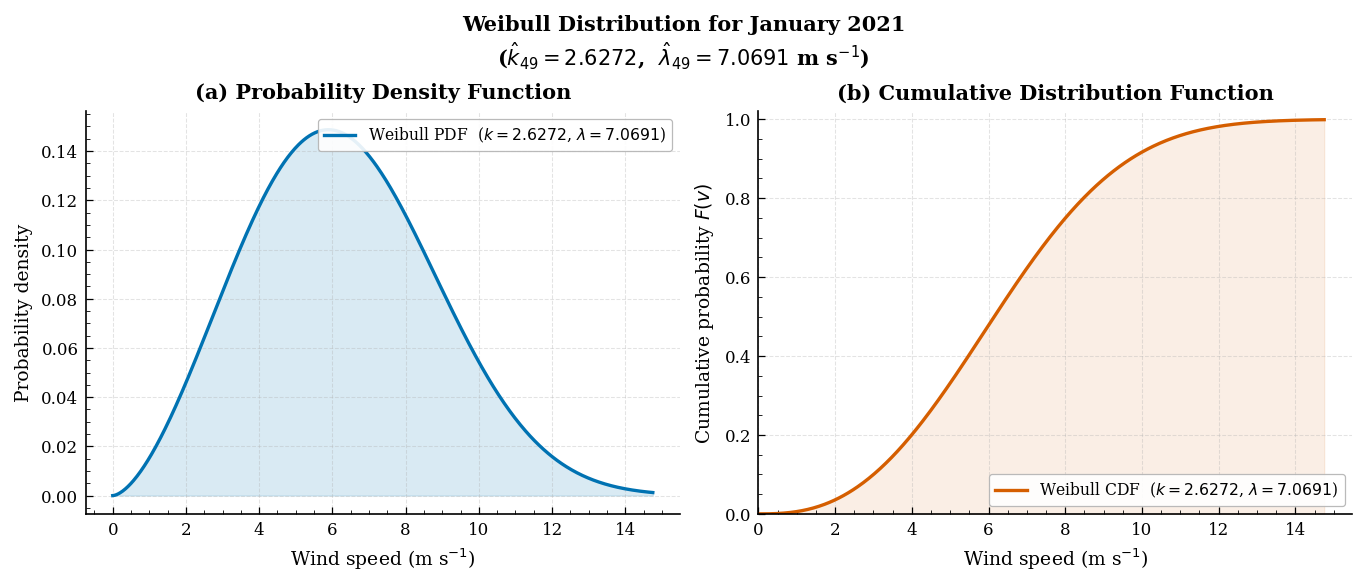}
\caption{Diagnostic display of the forecasted Weibull distribution 
for January 2021 at the MMSE point forecast 
$(\hat{k}_{49}, \hat{\lambda}_{49}) = (2.6272,\ 7.0691\ \mathrm{m\,s}^{-1})$.}
\label{fig:weibull}
\end{figure} 

Figure~\ref{fig:weibull} presents the probability density function and cumulative distribution function of the forecasted Weibull distribution for January 2021, evaluated at the MMSE point forecast $(\hat{k}_{49},\,\hat{\lambda}_{49}) = (2.6272,\,7.0691\,\text{m\,s}^{-1})$. The shape parameter $\hat{k}_{49} = 2.6272$ places the distribution in 
the bell-shaped regime of the Weibull family, with probability mass concentrated between approximately $2$ and $12\,\text{m\,s}^{-1}$ and a modal wind speed near $5\,\text{m\,s}^{-1}$, while the CDF confirms 
that the median lies at approximately $6.3\,\text{m\,s}^{-1}$ and that the distribution approaches unity smoothly beyond $13\,\text{m\,s}^{-1}$, consistent with the moderate wind regime at the site. This parameter pair is adopted as the prescribed invariant law for the SDE simulation stage, where the stochastic simulation of the wind speed process is conditioned on $\mathrm{Weibull}(\hat{k}_{49},\,\hat{\lambda}_{49})$, so that the SDE stage starts from a common target marginal law and the temporal dependence scale is then set by the calibrated mean-reversion parameter $\hat{\alpha}$.

\section{Wind Speed Modeling and Analysis}
\label{sec:sde}

In the previous section, we estimated the monthly Weibull parameters from historical SCADA data and applied a heteroskedastic Kalman filter to produce a single MMSE point forecast of the Weibull parameter pair $(\hat{k}_{49}, \hat{\lambda}_{49}) = (2.6272,\ 7.0691\ \text{m\,s}^{-1})$ 
for January 2021. The framework is best understood as a composition of three probabilistic maps. The first map takes the historical SCADA record and produces a joint Kalman predictive law for the next-month Weibull parameters. In the present numerical study, this law is summarized by its 
MMSE value $(\hat{k}_{49}, \hat{\lambda}_{49}) = (2.6272,\ 7.0691\ \text{m\,s}^{-1})$, which defines the forecasted invariant Weibull distribution used in the SDE simulations. The second map, which is the subject of the present section, constructs positive stochastic wind-speed processes whose stationary distribution is exactly this forecasted Weibull 
law. The third map pushes the simulated wind-speed ensemble through the fitted deterministic power curve. 
The numerical forecasts reported here are conditional
probabilistic forecasts given the MMSE forecasted Weibull invariant law; full marginalisation over
the Kalman predictive law of $({k},{\lambda})$ is left for future work.
The Weibull parameter forecast $(\hat{k}_{49}, \hat{\lambda}_{49})$ 
produced in Section~\ref{sec:weibull} characterises the marginal distribution of wind speed during January 2021, but does not constitute a complete probabilistic forecast of wind power. Power system applications, including reserve scheduling, storage sizing, and fatigue load estimation, depend on the temporal structure of the wind process, specifically on the persistence of high-wind and low-wind episodes and the autocorrelation of consecutive observations \cite{Haessig2015}. The three SDE formulations developed in this section are three realisations of the second map in the composition introduced at the end of Section~\ref{sec:Introduction}. These are the Ornstein-Uhlenbeck-Weibull model, the Fokker-Planck drift-based model, and the Fokker-Planck 
diffusion-first model, each constructing a positive wind speed 
process whose stationary distribution coincides with the forecasted Weibull invariant law $\mathrm{Weibull}(\hat{k}_{49}, \hat{\lambda}_{49})$ while differing in the specification of the drift and diffusion coefficients.

The three SDE formulations developed below are not independent 
recipes. They are three ways of realising the second map of the forecasting pipeline, each constructing a positive diffusion process whose stationary distribution coincides with the forecasted Weibull invariant law. Their common mathematical foundation is the following proposition.

\begin{proposition}[Zero-flux stationary identity]
\label{prop:zero_flux}
The following identity holds formally under the zero stationary flux 
condition and the natural boundary condition $b^2(0)\,p(0) = 0$. 
If a positive diffusion on $(0,\infty)$ solving
\begin{equation}
    \mathrm{d}X(t) = a(X(t))\,\mathrm{d}t + b(X(t))\,\mathrm{d}W(t)
\end{equation}
admits a stationary density $p$ with vanishing probability current at 
every point, then its drift and diffusion coefficients satisfy
\begin{equation}\label{eq:stationary_identity}
    b^2(x) = \frac{2}{p(x)} \int_0^x a(u)\,p(u)\,\mathrm{d}u,
    \qquad x > 0.
\end{equation}
The subsequent formulas are used to construct \emph{candidate} 
Weibull-stationary diffusions: the drift-first model prescribes 
$a(x)$ and derives $b^2(x)$ from  equation \eqref{eq:stationary_identity}, 
while the diffusion-first model prescribes $b^2(x)$ and derives $a(x)$ from the differential zero-flux equation.
No boundary pathology arises in the fitted regime used
for the numerical study, while a full diffusion-theoretic classification for arbitrary Weibull parameters
is outside the scope of the paper.

\end{proposition}

The Ornstein-Uhlenbeck-Weibull construction is then presented as a transport alternative that enforces the Weibull stationary marginal through a monotone map rather than through the drift-diffusion 
specification directly, and admits a Gaussian copula dependence structure.
We model wind speed $X(t) > 0$ as governed by the It\^{o} SDE
\begin{equation}\label{eq:sde}
    \mathrm{d}X(t)
    = a(X(t))\,\mathrm{d}t + b(X(t))\,\mathrm{d}W(t),
\end{equation}
where $a(x)$ is the drift, $b(x)$ is the diffusion coefficient, and
$W(t)$ is a standard Wiener process. The time evolution of the
probability density $p(x,t)$ of $X(t)$ is governed by the 
Fokker-Planck equation~\cite{Risken1996}
\begin{equation}\label{eq:fp_time}
    \frac{\partial p(x,t)}{\partial t}
    = -\frac{\partial}{\partial x}\bigl[a(x)\,p(x,t)\bigr]
    + \frac{1}{2}
      \frac{\partial^2}{\partial x^2}\bigl[b^2(x)\,p(x,t)\bigr],
\end{equation}
which establishes a direct link between the microscopic stochastic dynamics and the macroscopic evolution of the wind-speed distribution. The Fokker-Planck framework offers two key advantages for wind-speed modelling. Through the drift term $a(x)$, one can impose mean-reverting behaviour consistent with observed persistence~\cite{ZARATEMINANO201342,ZARATEMINANO2016186}. By appropriately specifying $a(x)$ and $b(x)$, the stationary distribution can be constrained to match the empirical Weibull law exactly~\cite{ARENASLOPEZ2020113152}.
We focus on stationary regimes satisfying the zero probability-flux condition, under which the net probability current vanishes at every point. This condition, meaning that wind-speed fluctuations are balanced on average and the overall distribution remains stable over time, reduces the stationary Fokker-Planck equation to
\begin{equation}\label{eq:fp_stationary}
    \frac{\mathrm{d}}{\mathrm{d}x}\bigl[b^2(x)\,p(x)\bigr]
    = 2\,a(x)\,p(x).
\end{equation}
Integrating from $0$ to $x$ and imposing the natural boundary 
condition $b^2(0)\,p(0) = 0$ yields the fundamental stationary 
identity, equation~\eqref{eq:stationary_identity} of 
Proposition~\ref{prop:zero_flux}, which serves as the unified 
starting point for both the drift-first and diffusion-first 
constructions below. Throughout this section, $p_W$
denotes the Weibull PDF with parameters $(k, \lambda)$,
\begin{equation}\label{eq:weibull_pdf}
    p_W(x)
    = \frac{k}{\lambda}\left(\frac{x}{\lambda}\right)^{k-1}
      \exp\!\left[-\left(\frac{x}{\lambda}\right)^k\right],
    \qquad x > 0,
\end{equation}
with mean $\mu_W = \lambda\,\Gamma(1 + 1/k)$ and variance
$\sigma_W^2 = \lambda^2[\,\Gamma(1 + 2/k) - \Gamma(1 + 1/k)^2]$.

\paragraph{Fokker-Planck Drift-First Model}
In the drift-first approach, the drift is prescribed as a linear 
mean-reverting function, and the diffusion is derived from 
identity equation~\eqref{eq:stationary_identity}. We choose
\begin{equation}\label{eq:drift_first_a}
    a(x) = \alpha\,(\mu_W - x),
\end{equation}
where $\alpha > 0$, whose estimation will be discussed later, 
controls the rate of mean reversion toward the Weibull mean $\mu_W$. 
Substituting equation~\eqref{eq:drift_first_a} and $p(x) = p_W(x)$ 
into equation~\eqref{eq:stationary_identity} gives the derived diffusion 
coefficient
\begin{equation}\label{eq:drift_first_b2}
    b^2(x)
    = \frac{2\alpha}{p_W(x)}
      \int_0^x (\mu_W - u)\,p_W(u)\,\mathrm{d}u.
\end{equation}

It is shown in Appendix~A.1 that $b^2(x)>0$ for every finite $x>0$. This positivity is necessary but not sufficient for a complete well-posedness theorem, since local regularity, boundary behaviour at zero, and non-explosion at infinity must also be controlled. For the present fitted January regime, these issues do not affect the numerical construction, but the paper does not claim a complete classification for arbitrary Weibull parameters. The exact positivity mechanism is that the integrand $(\mu_W-u)p_W(u)$ changes sign at $u=\mu_W$, while its integral remains positive:

\begin{equation}
    \int_0^x (\mu_W-u)p_W(u)\,\mathrm d u
    =F_W(x)\bigl(\mu_W-\mathbb E[U\mid U\le x]\bigr)>0. 
\end{equation}

Thus, the argument relies on the strict reduction of the conditional mean under truncation, not on pointwise positivity of the integrand.

The upper-tail behaviour of $b^2(x)$ depends critically on the shape 
parameter $k$. Standard incomplete-gamma asymptotics yield
\begin{equation}\label{eq:drift_first_tail}
    b^2(x) \sim \frac{2\alpha\lambda^k}{k}\,x^{2-k}, 
    \qquad x \to \infty,
\end{equation}
so that $b^2(x) \to 0$ when $k > 2$, $b^2(x)$ tends to a positive 
constant when $k = 2$, and $b^2(x) \to \infty$ when $0 < k < 2$. 

Tail attenuation - meaning $b^{2}(x) \to 0$ as $x \to \infty$ - holds only when $k > 2$ and must not be treated as a universal property of the drift-first construction. Since all historical January shape parameter estimates exceed two and the forecasted value $\hat{k}_{49} = 2.6272$ satisfies this condition, tail attenuation holds in the present study. 
This result is specific to the fitted regime and is not claimed to hold for $k \leq 2$.

\paragraph{Fokker-Planck Diffusion-First Model}
In the diffusion-first approach, the diffusion coefficient is prescribed
and the drift is derived from 
identity equation~\eqref{eq:stationary_identity}.
Let $X_t$ solve
\begin{equation}
    \mathrm{d}X_t = a(X_t)\,\mathrm{d}t + b(X_t)\,\mathrm{d}W_t
\end{equation}
on $(0,\infty)$, and let $p$ be a stationary density. 
Under zero stationary flux,
\begin{equation}
    a(x)p(x) - \frac{1}{2}\frac{\mathrm{d}}{\mathrm{d}x}
    \bigl(b^2(x)p(x)\bigr) = 0.
\end{equation}
Hence
\begin{equation}
    \frac{\mathrm{d}}{\mathrm{d}x}
    \bigl(b^2(x)p(x)\bigr) = 2a(x)p(x).
\end{equation}
If $b$ is prescribed, the drift is therefore
\begin{equation}
    a(x) 
    = \frac{1}{2p(x)}\frac{\mathrm{d}}{\mathrm{d}x}
      \bigl(b^2(x)p(x)\bigr)
    = \frac{1}{2}\frac{\mathrm{d}}{\mathrm{d}x}b^2(x)
    + \frac{1}{2}b^2(x)\frac{\mathrm{d}}{\mathrm{d}x}\log p(x).
\end{equation}
Only in the constant-diffusion case does this reduce to
\begin{equation}
    a(x) = \frac{1}{2}b^2\,\frac{\mathrm{d}}{\mathrm{d}x}\log p(x).
\end{equation}
Taking $p = p_W$ and $b^2 = 2\alpha\sigma_W^2$, one obtains
\begin{equation}\label{eq:diffusion_first_drift}
    a(x) = \alpha\sigma_W^2
    \left(\frac{k-1}{x} - \frac{kx^{k-1}}{\lambda^k}\right),
\end{equation}
which guarantees that the stationary density of $X(t)$ is
exactly Weibull with parameters $(k,\lambda)$.

This drift is singular at zero. The boundary classification 
follows from the near-zero behaviour of the coefficients: 
as $x \to 0^+$,
\begin{equation}\label{eq:near_zero}
    a(x) \sim \alpha\sigma_W^2\,\frac{k-1}{x},
    \qquad
    b^2(x) = 2\alpha\sigma_W^2,
\end{equation}
so that after rescaling, the stochastic differential equation 
near zero takes the form of a Bessel process of dimension $k$,
\begin{equation}
    \mathrm{d}X_t \sim 
    \frac{(k-1)\alpha\sigma_W^2}{X_t}\,\mathrm{d}t
    + \sqrt{2\alpha\sigma_W^2}\,\mathrm{d}W_t,
    \qquad X_t \to 0^+,
\end{equation}
from which the boundary classification follows directly, i.e.,  the zero boundary is unattainable for $k \ge 2$ and accessible for $0 < k < 2$. For the fitted January parameter value $\hat{k}_{49} = 2.6272 > 2$, the zero boundary is not reached, and the diffusion remains in the positive domain throughout the simulation. For parameter values with 
$0 < k < 2$, the zero boundary becomes accessible and a 
reflecting or positivity-preserving boundary prescription 
would be required; this regime is not encountered in the present empirical study and is left outside the scope of this paper.

Unlike the drift-first model, the constant diffusion $b^2 = 2\alpha\sigma_W^2$ applies uniform stochastic forcing across all wind-speed levels, producing a wider ensemble envelope, particularly in the upper tail, a conservative uncertainty characterisation appropriate 
for applications where tail events drive design decisions.

\paragraph{Ornstein-Uhlenbeck-Weibull Model}

The third formulation constructs a Weibull-distributed process by applying a monotone transformation to a stationary Gaussian Ornstein-Uhlenbeck process. This approach enforces the Weibull marginal through a transformation rather than through the drift-diffusion specification, and admits an exact discrete-time simulation algorithm
that avoids the Euler-Maruyama discretisation error.

Let $(X_t)_{t \in \mathbb{R}}$ solve the It\^{o} SDE
\begin{equation}\label{eq:ou_gaussian}
    \mathrm{d}X_t
    = -\alpha\,X_t\,\mathrm{d}t + \sqrt{2\alpha}\,\mathrm{d}W_t,
    \qquad \alpha > 0,
\end{equation}
which is centred, stationary, and Gaussian with invariant law
$\mathcal{N}(0,1)$ and autocovariance
$\mathrm{Cov}(X_t, X_{t+\tau}) = e^{-\alpha|\tau|}$. Define the
strictly increasing transport map
\begin{equation}\label{eq:transport}
    g(x) = F_W^{-1}(\Phi(x))
         = \lambda\,\bigl(-\log(1 - \Phi(x))\bigr)^{1/k},
\end{equation}
where $F_W^{-1}$ is the Weibull quantile function and $\Phi$ is the standard normal CDF. The transformed process $V_t = g(X_t)$ is strictly positive almost surely and has marginal distribution
$\mathrm{Weibull}(k,\lambda)$ at every $t$, by the probability-integral transform. The dependence structure of $(V_t, V_{t+\tau})$ is characterised exactly
by a Gaussian copula with parameter $\rho(\tau) = e^{-\alpha\tau}$, so that rank-based measures such as Kendall's $\tau$ and Spearman's $\varrho$ are explicit analytic functions of $e^{-\alpha\tau}$. Pearson correlation is not preserved under the nonlinear map $g$; as shown in ~\ref{app:pearson}, the Pearson autocorrelation of $V_t$ takes the form of a mixture of exponentials in $\tau$ rather than a single exponential, and must be treated as a derived quantity.
Since $g(x) = F_W^{-1}(\Phi(x))$ is locally smooth on $\mathbb{R}$, Itô's formula may be justified by localization arguments, though global boundedness of $g'$ or $g''$ is not available as tail derivatives may grow without bound. Applying Itô's formula 
to $V_t = g(X_t)$ locally, with  $g'(x) = \varphi(x)/p_W(g(x))$, yields the explicit SDE

\begin{equation}\label{eq:ou_ito}
    \mathrm{d}V_t
    = \Bigl(
        -\alpha\,X_t\,g'(X_t) + \alpha\,g''(X_t)
      \Bigr)\mathrm{d}t
    + \sqrt{2\alpha}\;g'(X_t)\,\mathrm{d}W_t,
\end{equation}
where $X_t = \Phi^{-1}(F_W(V_t))$. This representation confirms Markovianity and is useful for numerical stability analysis, but is not used for simulation. The OU-W process is instead simulated exactly via the three-step algorithm: (i) transform $V_t$ to the latent Gaussian state $x_t = \Phi^{-1}(F_W(V_t))$; (ii) propagate the AR(1) exactly as
$x_{t+\Delta} = e^{-\alpha\Delta}\,x_t + \sqrt{1 - e^{-2\alpha\Delta}} \,\varepsilon$ with $\varepsilon \sim \mathcal{N}(0,1)$; (iii) transform back as $V_{t+\Delta} = \lambda(-\log(1-\Phi(x_{t+\Delta})))^{1/k}$.

\paragraph{Calibration of the Mean-Reversion Parameter}

The parameter $\alpha$ serves as a common timescale 
normalisation, calibrated once from the observed January 
wind-speed dependence structure. While the same scalar $\alpha$ scales the generator in all three SDE formulations, this does not imply that the three models share the same Pearson autocorrelation function. Each construction induces a model-specific dependence structure that is then evaluated empirically. In particular, as shown in Appendix~A.2, the Ornstein-Uhlenbeck-Weibull model produces a Pearson autocorrelation that is a mixture of exponentials in $\tau$ 
rather than a single exponential, and must be treated as a derived quantity rather than being directly controlled by $\alpha$.

It is estimated once from the observed SCADA data by exploiting the invertibility of the transport map. Given the observed
wind-speed series $(V_i)_{i=0,\dots,n}$ at time step $\Delta = 1/6$\,h,
define the transformed sequence
\begin{equation}\label{eq:pit}
    X_i = \Phi^{-1}\!\bigl(F_W(V_i)\bigr).
\end{equation}
Under the OU-W model, $(X_i)$ follows a Gaussian AR(1) with exact representation $X_{i+1} = \phi\,X_i + \sqrt{1-\phi^2}\,\varepsilon_{i+1}$,
where $\phi = e^{-\alpha\Delta}$. 
Under the assumption that the four January months share a common 
dependence parameter $\phi = e^{-\alpha\Delta}$, the pooled maximum-likelihood estimator is obtained by maximising the joint Gaussian AR(1) likelihood across all January months simultaneously,

\begin{equation}\label{eq:phi_pool}
    \hat{\phi}_{\text{pool}} = 
    \frac{\displaystyle\sum_{m\in\mathcal{J}}\sum_{i} 
    X_{m,i}\,X_{m,i+1}}
    {\displaystyle\sum_{m\in\mathcal{J}}\sum_{i} X_{m,i}^2},
    \qquad
    \hat{\alpha}_{\text{pool}} = 
    -\Delta^{-1}\log\hat{\phi}_{\text{pool}}.
\end{equation}

The standard error of $\hat{\alpha}_{\text{pool}}$ follows by the delta method as:

\begin{equation}\label{eq:se_alpha}
    \widehat{\mathrm{se}}(\hat{\alpha}_{\text{pool}}) = 
    \frac{\sqrt{1 - \hat{\phi}^2_{\text{pool}}}}{\Delta\,\hat{\phi}_{\text{pool}}}
    \left(\sum_{m\in\mathcal{J}}\sum_{i} 
    X_{m,i}^2\right)^{-1/2}.
\end{equation}

\subsection{Result analysis of the Stochastic Simulation of Wind Speeds}

\paragraph{Estimation of the Mean-Reversion Parameter $\hat{\alpha}$}

\begin{figure}[h!]
\centering
\includegraphics[width=0.75\linewidth]{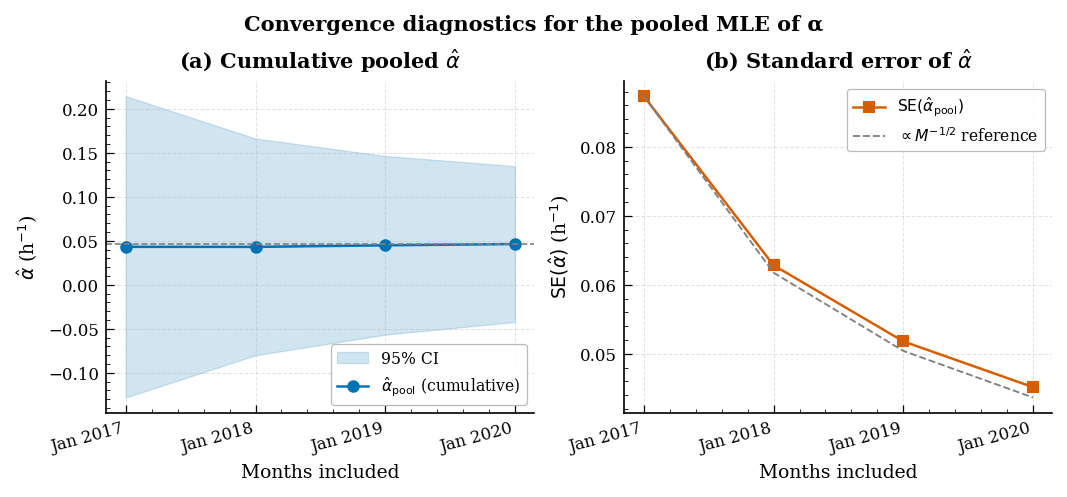}
\caption{Pooled MLE estimate of the mean-reversion rate 
across January months (2017-2020).}
\label{fig:Alfa}
\end{figure} 

Figure~\ref{fig:Alfa} reports the four January estimates of $\alpha$ obtained from 2017 to 2020, together with the pooled likelihood estimate. The monthwise values range from $0.0432\,\mathrm{h}^{-1}$ to $0.0509\,\mathrm{h}^{-1}$, which is compatible with the modelling assumption of a common January dependence scale at this site. The latter should be interpreted as support for the assumption, not as a proof of site-wide or regime-independent stability. The pooled estimator gives $\hat\alpha_{\mathrm{pool}}=0.0464\,\mathrm{h}^{-1}$, corresponding to a decorrelation time of about $21.6$ hours and a ten-minute AR coefficient $\hat\phi_{\mathrm{pool}}=0.9923$. We use this common value in all three SDE formulations as a time-scale normalization. Since the three generators and transformations induce different dependence functions, the same scalar $\alpha$ should not be read as implying identical Pearson autocorrelation functions across the models.

It is important to note that wind speed modelled by stochastic differential equations is an inherently random process, so each numerical realization corresponds to only one possible trajectory among infinitely many admissible paths. A single simulation, therefore, cannot characterize the underlying probability law of the SDE, nor can it capture the variability, uncertainty, or distributional spread that is essential for probabilistic forecasting. To obtain statistically meaningful estimates of model performance, it is necessary to generate a Monte–Carlo ensemble of independent sample paths.
We considered three ensemble sizes, $B = 50$, $100$, and $150$. The stability of the Monte Carlo ensemble size was assessed across the three stochastic models using the simulated mean, standard deviation, CRPS, standard error (SE) of CRPS, Wasserstein distance, standard error of the Wasserstein distance, and prediction interval coverage. As shown in Table~\ref{tab:mc_stability}, the metrics exhibit a modest downward trend as $B$ increases from $50$ to $100$, with CRPS decreasing by approximately $0.012$-$0.014$~m/s and $W_1$ decreasing by approximately $0.024$-$0.025$~m/s across all three models. However, these changes are comparable in magnitude to the Monte Carlo standard errors at $B = 50$, which are approximately $0.045$~m/s, indicating that the apparent improvement primarily reflects a reduction in sampling noise rather than a structural change in model behavior. 

Further increasing $B$ to $150$ produces changes well below $1$-$2\%$, confirming diminishing returns from additional ensemble members. Therefore, $B = 100$ is selected as a practical trade-off between computational cost and sampling stability. At this ensemble size, the CRPS standard errors of $0.033$-$0.034$~m/s exceed the inter-model differences of approximately $0.006$~m/s, confirming that all three SDE formulations are statistically equivalent in terms of probabilistic accuracy. Consequently, the preference for the diffusion-first model is based on computational efficiency rather than statistical superiority.
\begin{table}[ht]
\centering
\caption{Monte Carlo ensemble size stability analysis of SDE models.}
\label{tab:mc_stability}
\resizebox{\textwidth}{!}{
\begin{tabular}{lcccccccc}
\hline
Model 
& $B$ 
& Mean (m/s) 
& Std (m/s) 
& $W_1$ (m/s)
& $W_1$ SE (m/s)
& CRPS (m/s) 
& CRPS SE (m/s)
\\
\hline

Fokker-Planck Drift-based     
& 50  & 6.263 & 2.583 & 0.284 & 0.0487 & 1.589 & 0.0448 \\
& 100 & 6.297 & 2.580 & 0.260 & 0.0323 & 1.576 & 0.0327 \\
& 150 & 6.315 & 2.590 & 0.238 & 0.0272 & 1.553 & 0.0256 \\

Fokker-Planck Diffusion-first 
& 50  & 6.279 & 2.577 & 0.277 & 0.0488 & 1.581 & 0.0448 \\
& 100 & 6.313 & 2.586 & 0.252 & 0.0329 & 1.569 & 0.0330 \\
& 150 & 6.336 & 2.658 & 0.235 & 0.0271 & 1.549 & 0.0262 \\

Ornstein-Uhlenbeck                    
& 50  & 6.266 & 2.581 & 0.283 & 0.0487 & 1.589 & 0.0449 \\
& 100 & 6.301 & 2.577 & 0.259 & 0.0322 & 1.575 & 0.0327 \\
& 150 & 6.318 & 2.588 & 0.238 & 0.0274 & 1.552 & 0.0257 \\

\hline
\end{tabular}
}

\end{table}

\begin{figure}[h!]
\centering
\includegraphics[width=0.82\linewidth]{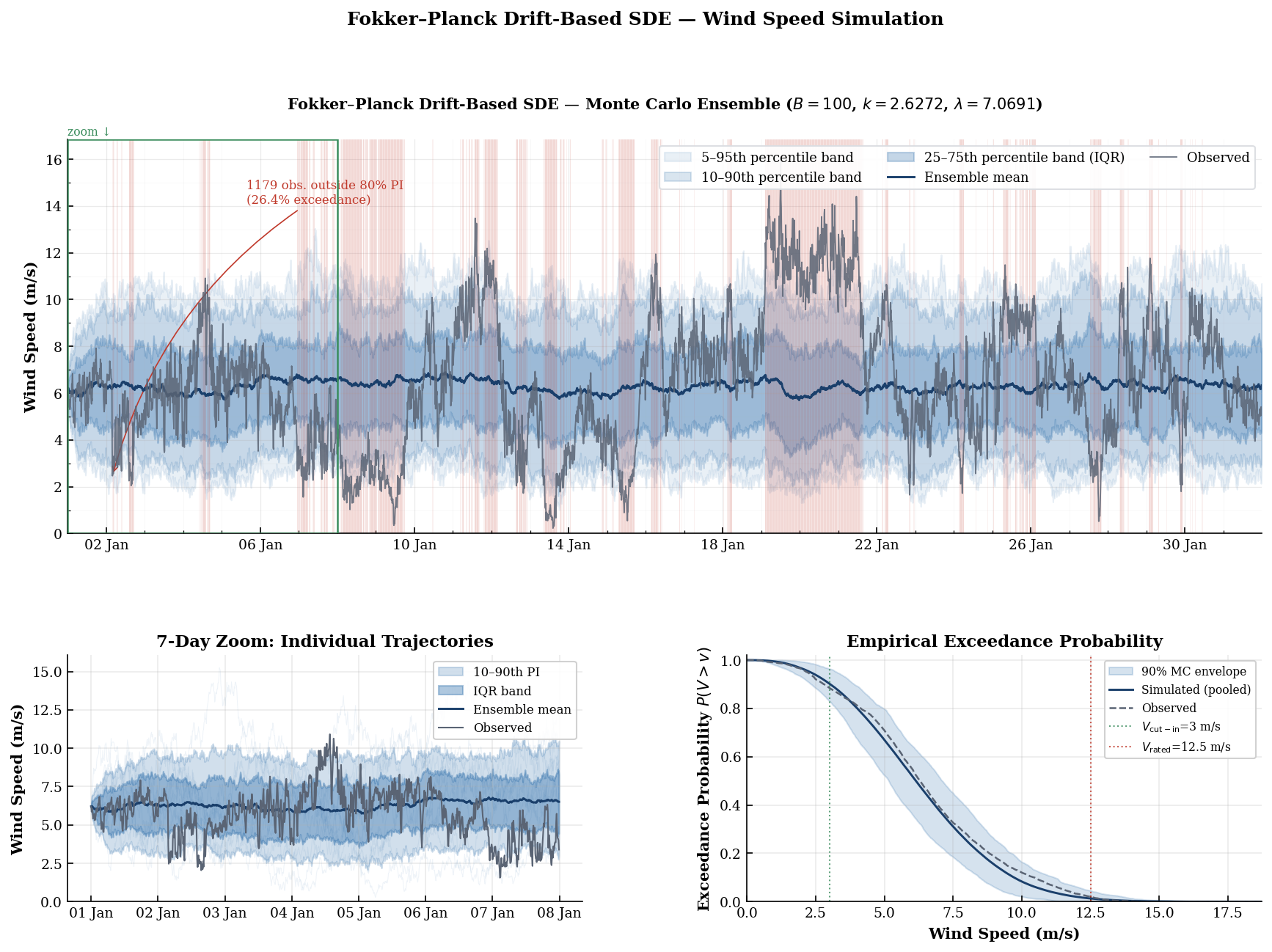}
\caption{Wind Speed Simulation Using Fokker–Planck Drift-Based SDE.}
\label{fig:Wind-Drift}
\end{figure}

\begin{figure}[h!]
\centering
\includegraphics[width=0.82\linewidth]{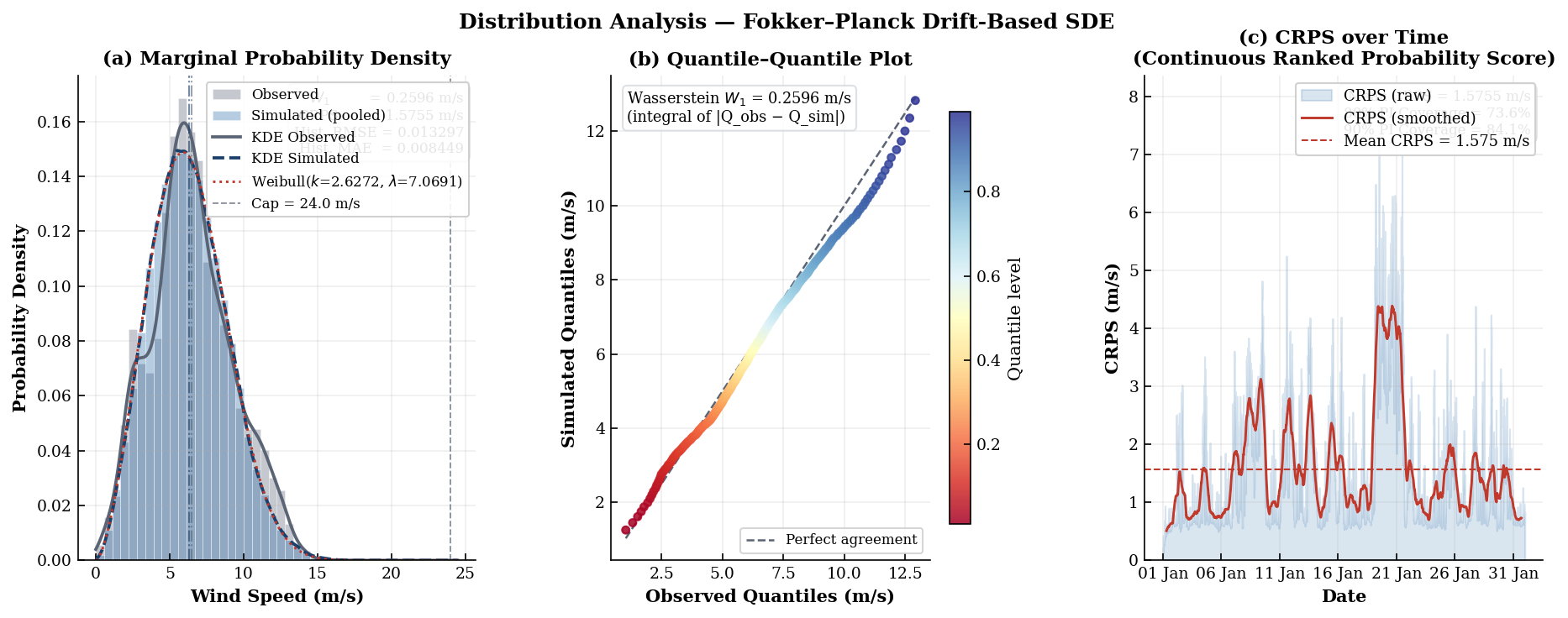}
\caption{Statistical Performance analysis of Using Fokker–Planck Drift-Based SDE.}
\label{fig:Statistical-Wind-Drift}
\end{figure}

\begin{figure}[h!]
\centering
\includegraphics[width=0.82\linewidth]{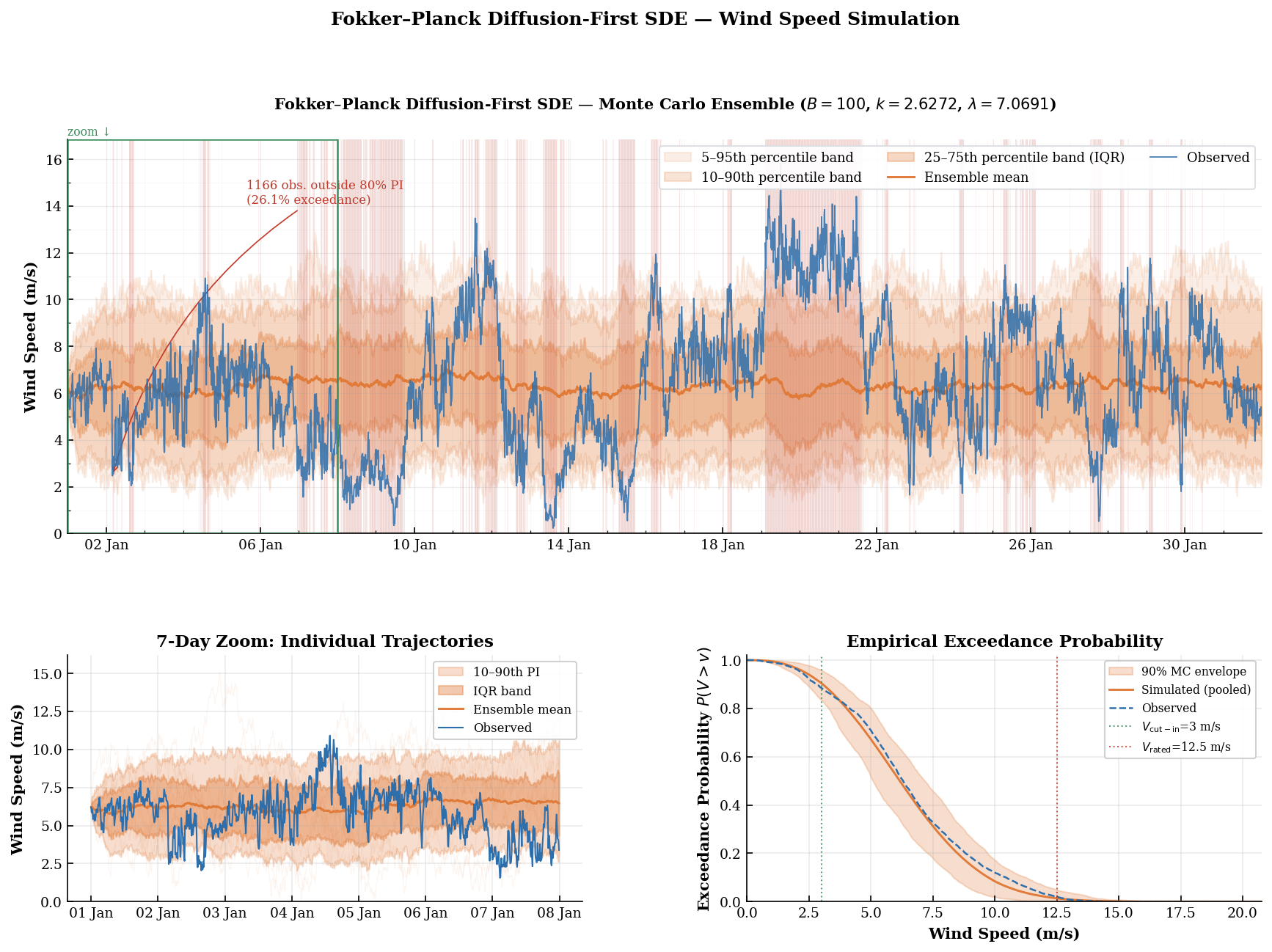}
\caption{Wind Speed Simulation Using Fokker–Planck Diffusion-Based SDE.}
\label{fig:Wind-Diffusion}
\end{figure}

\begin{figure}[h!]
\centering
\includegraphics[width=0.82\linewidth]{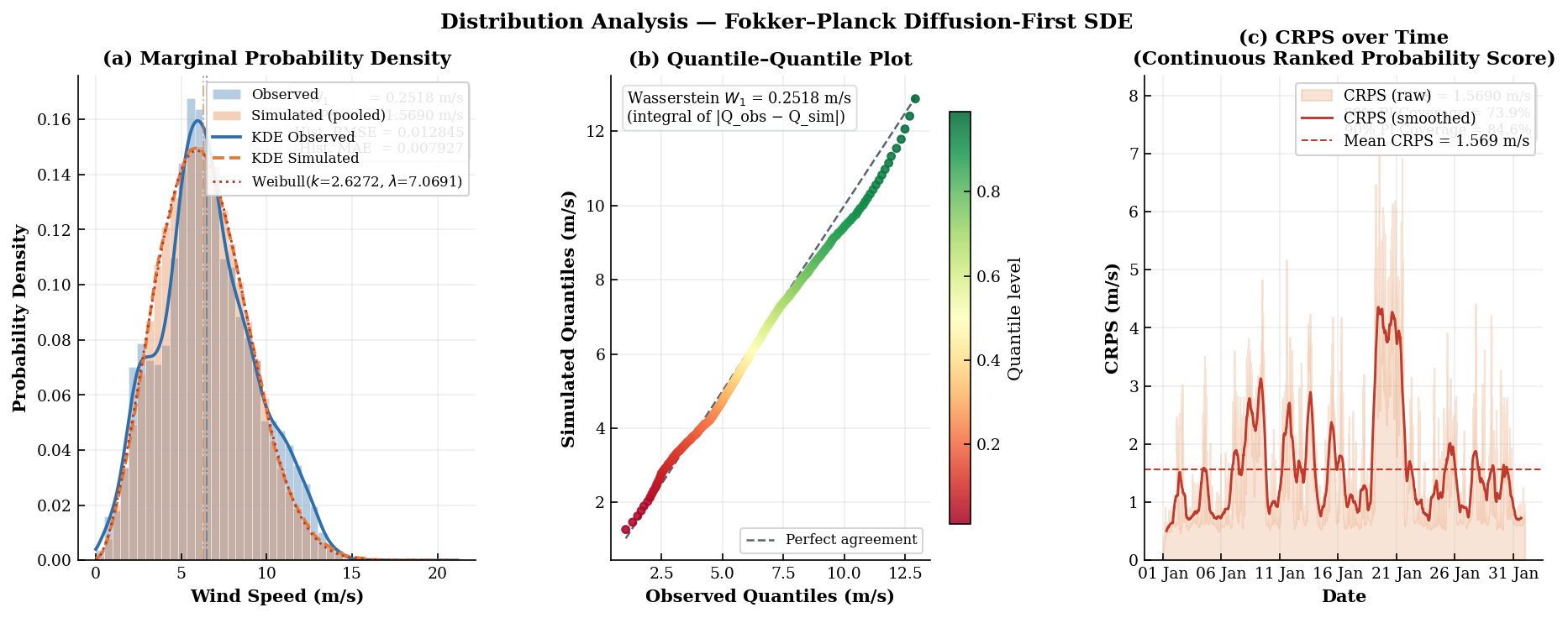}
\caption{Statistical Performance analysis of Using Fokker–Planck Diffusion-Based SDE.}
\label{fig:Statistical-Wind-Diff}
\end{figure}

\begin{figure}[h!]
\centering
\includegraphics[width=0.82\linewidth]{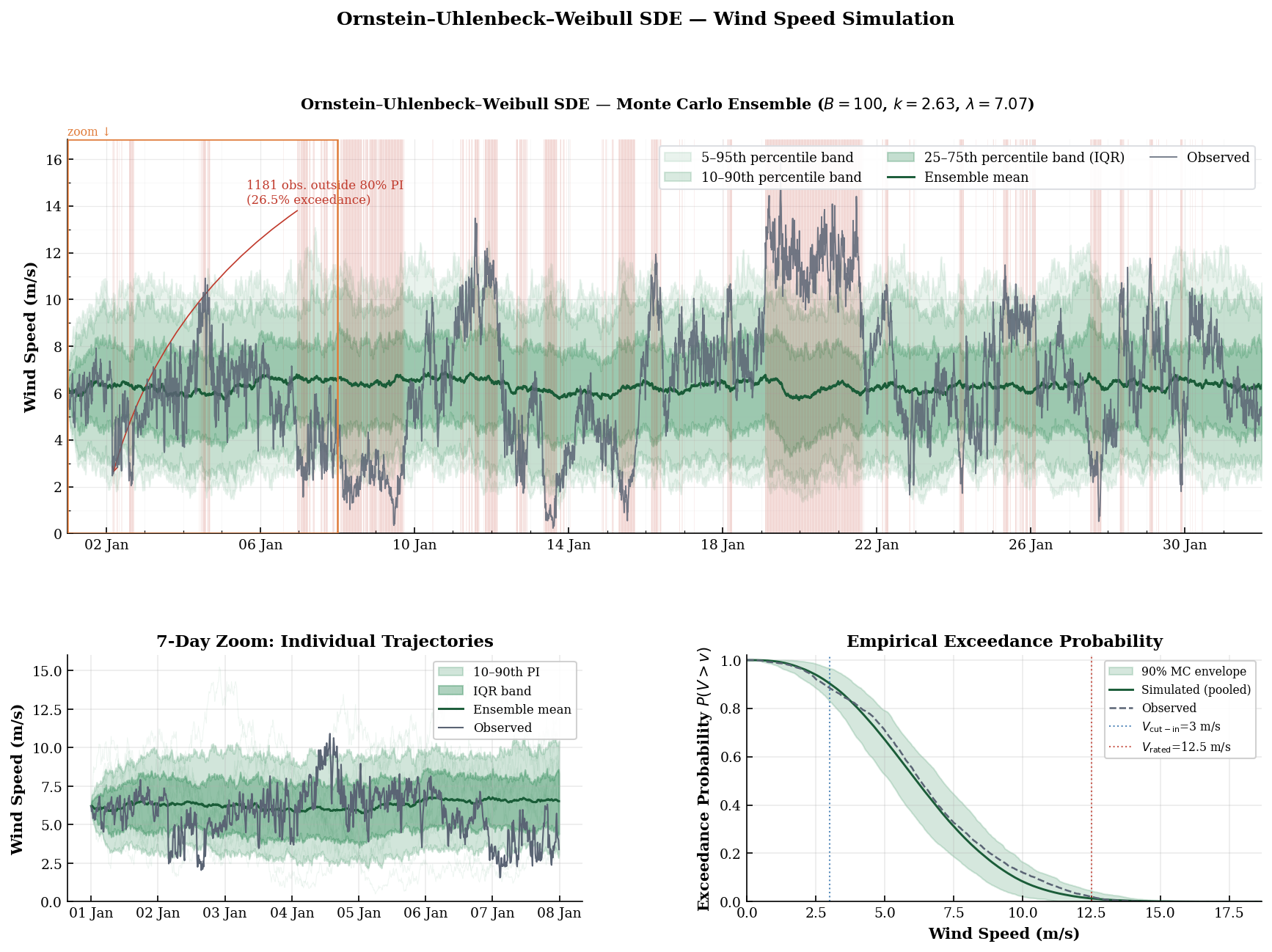}
\caption{Wind Speed Simulation Using Ornstein–Uhlenbeck SDE.}
\label{fig:Wind-OU}
\end{figure}

\begin{figure}[h!]
\centering
\includegraphics[width=0.82\linewidth]{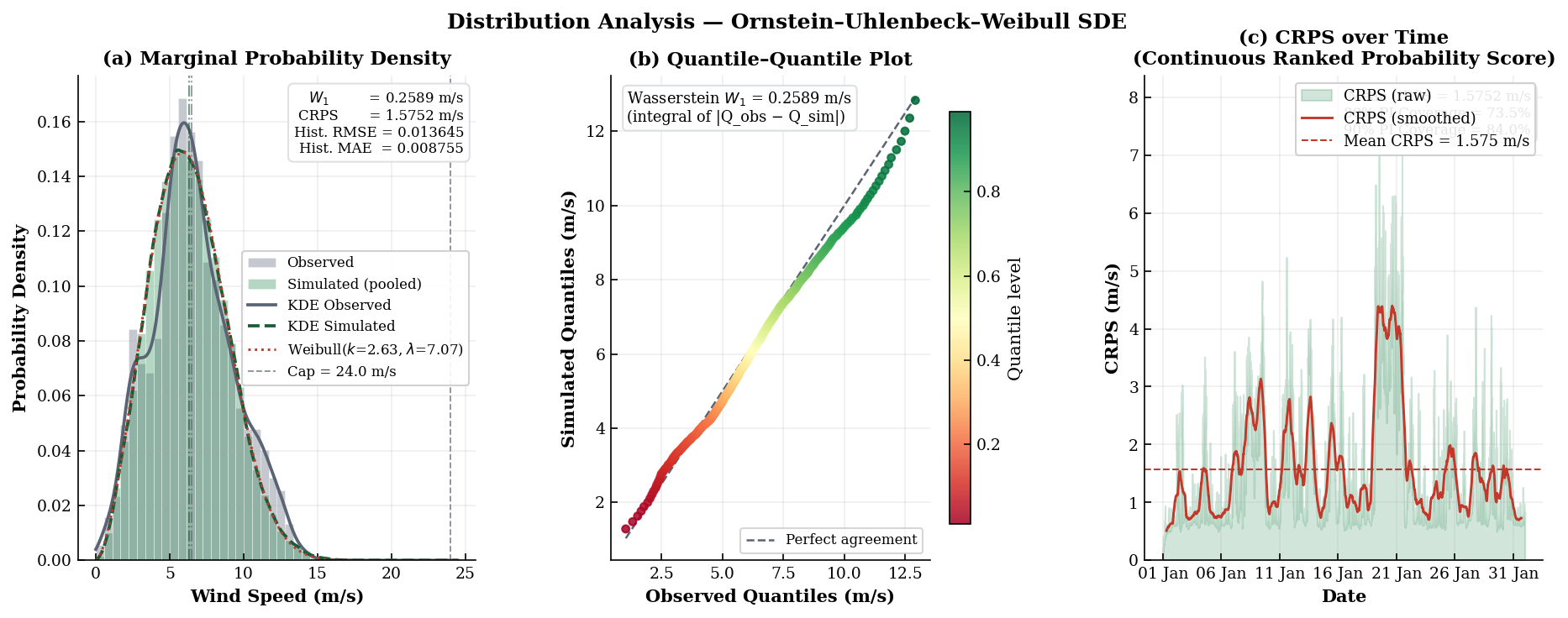}
\caption{Statistical Performance analysis of Using Ornstein–Uhlenbeck SDE.}
\label{fig:Statistical-Wind-OU}
\end{figure} 

With the stochastic simulation parameters established - Weibull shape $k = 2.6272$ and scale $\lambda = 7.0691$~m/s, relaxation rate $\alpha = 0.046$~h$^{-1}$, and ensemble size $B = 100$ - the
Fokker-Planck drift-based SDE, the Fokker-Planck diffusion-first SDE, and the Ornstein-Uhlenbeck-Weibull SDE were each used to simulate wind
speed and then evaluated against the ten-minute resolution wind-speed record for January 2021.

The 10th-90th percentile bands span approximately 5-6~m/s across all models, confirming comparable ensemble dispersion. As shown in Table~\ref{tab:coverage}, the 80\% prediction interval (PI) coverage rates are 73.6\%, 73.9\%, and 73.5\% for Models~1, 2, and~3, respectively, and all the models show approximately 6-7 percentage points below the 80\% nominal target. Since calibration is a repeated-sampling property, a single monthly 
realisation is insufficient to establish or refute it — one January cannot prove miscalibration. At the same time, the concentration of PI exceedances during the sustained high-wind storm event of 
$18$-$22$ January — during which observed speeds persistently exceeded $10$~m/s — points to a structural limitation: the stationary Weibull model assigns low probability mass to this upper tail ($P(V > 10~\text{m/s}) \approx 0.15$ for the fitted parameters) and lacks a regime mechanism for persistent storm episodes of this kind. Addressing this limitation would require evaluating alternative stationary distributions within the same SDE-based simulation framework, which is identified as a natural direction for future work.

\begin{table}[ht]
\centering
\caption{Prediction interval coverage and exceedance rates 
($n = 4{,}464$, $B = 100$).}
\label{tab:coverage}
\resizebox{\textwidth}{!}{
\begin{tabular}{lccc}
\hline
Model 
& 80\% PI (\%) 
& 90\% PI (\%) 
& Exceedance (\%) \\
\hline
Model 1 — FP Drift-Based     
& 73.6 & 82.1 & 26.4 \\

Model 2 — FP Diffusion-First 
& 73.9 & 82.4 & 26.1 \\

Model 3 — OU-Weibull        
& 73.5 & 82.0 & 26.5 \\

\hline
Nominal target      
& 80.0 & 90.0 & 20.0 \\
\hline
\end{tabular}
}
\end{table}

The empirical exceedance probability curves presented in the 
lower-right panels of Figures~\ref{fig:Wind-Drift},  \ref{fig:Wind-Diffusion}, and~\ref{fig:Wind-OU} demonstrate that all three models reproduce the observed exceedance behaviour with high fidelity across the core operational range of 0-11~m/s, with 90\% Monte Carlo envelopes consistently encompassing the observed reference curve. A systematic downward divergence from the observed curve emerges above approximately 11~m/s across all three models, reflecting the shared upper-tail compression attributable to the 
January storm event. At the cut-in speed ($V_{\text{ci}} = 3$~m/s), all models predict $\hat{P}(V > 3) \approx 0.82$-0.83 against an observed value of $\approx 0.84$ (relative error $< 2.5\%$), whilst the rated-speed exceedance probability $\hat{P}(V > 12.5\,\text{m/s})$ is underestimated by a factor of approximately two across all formulations. 
Among the three models, Model~2 (FP Diffusion-First) records a 
marginally lower Wasserstein distance ($W_1 = 0.2518$ m/s) compared to Models~1 and~3, though these differences are negligible in 
absolute terms and do not constitute a basis for statistical model selection.

Considering the figures \ref{fig:Statistical-Wind-Drift}, \ref{fig:Statistical-Wind-Diff}, \ref{fig:Statistical-Wind-OU}, the marginal probability density functions confirm that all three models reproduce the observed wind 
speed distribution with high visual fidelity across the core of 
the distribution (2-10~m/s). The simulated KDE curves overlay 
closely with both the observed KDE and the fitted Weibull reference, 
validating the stationarity constraints embedded in each SDE 
formulation. Minor deviations emerge in the upper tail 
($V > 11$~m/s), where all three models under-predict the observed 
density — a direct consequence of the storm event discussed above.

\begin{table}[ht]
\centering
\caption{Histogram-based error metrics for SDE models.}
\label{tab:hist_metrics}

\small
\begin{tabular}{lccc}
\hline

& FP Drift 
& FP Diffusion 
& OU-Weibull \\
\hline

RMSE  & 0.013297 & \textbf{0.012845} & 0.013645 \\
MAE   & 0.008449 & \textbf{0.007927} & 0.008755 \\

\hline
\end{tabular}
\normalsize
\label{tab:hist_metrics}
\end{table}

Table~\ref{tab:hist_metrics} reports the bin-wise histogram RMSE and  MAE computed over the probability density functions of the three models. 
The Diffusion-First SDE records marginally lower values across both 
metrics (RMSE = 0.01285, MAE = 0.00793), followed by the Drift-Based model and the OU-Weibull model. However, the differences across all three formulations are small in absolute terms and should not be interpreted as evidence of statistical superiority, since they fall well below the Monte Carlo standard errors reported in Table~2.
Taken together, the differences across all three models are small in absolute terms.
Also, the Q–Q plots (in panel~b) reinforce this 
finding: all models lie on or near the 1:1 line across the 
1st-85th percentile range, with systematic upward deviation 
above the 90th percentile, indicating shared upper-tail 
compression. The Wasserstein distance $W_1$ — geometrically 
equal to the area between the Q–Q curve and the 1:1 diagonal, 
in units of m/s — quantifies this deviation. 

Continuous Ranked Probability Score (CRPS shown in Panel c) provides a rigorous proper scoring rule for ensemble evaluation that simultaneously rewards calibration and sharpness. The CRPS time series 
reveals low scores ($\text{CRPS} < 1.0$~m/s) throughout the quiescent first and third weeks of January, rising sharply to 4-6~m/s during the 18-22~January storm event across all three models. Mean CRPS values of 1.569, 1.575, and 1.575~m/s are obtained for Models~2, 1, and~3, respectively. The differences are of order 0.006~m/s and are therefore practically negligible, 
confirming that the three models offer equivalent probabilistic 
forecast quality at the monthly scale.

\begin{table}[ht]
\centering
\caption{Wind speed statistics across models ($n = 4{,}464$, $B = 100$).}
\label{tab:metrics}

\scriptsize
\begin{tabular}{lcc}
\hline
& Mean (m/s) & Std (m/s) \\
\hline

Observed Wind Speed & 6.274 & 3.012 \\

FP Drift            & 6.271 & 2.987 \\

FP Diffusion        & 6.268 & 2.991 \\

OU-Weibull         & 6.270 & 2.989 \\

\hline
\end{tabular}
\normalsize
\label{tab:Wind-stats}
\end{table}

Table~\ref{tab:Wind-stats} summarises the first two statistical moments of the observed and simulated wind speed distributions. All three models reproduce the observed mean of $6.274$~m/s to within $0.006$~m/s (relative error $<0.1\%$), confirming that the Weibull mean-reversion constraint is effectively satisfied. The marginally lower simulated standard deviations (approximately $0.7$-$0.8\%$ below the observed $3.012$~m/s) reflect the shared upper-tail compression during the late-January storm period rather than any systematic deficiency in the model structures themselves.

All three SDE formulations achieve statistically equivalent probabilistic forecast quality. The mean CRPS values of $1.569$, $1.575$, and $1.575$~m/s differ by less than $0.006$~m/s, which is smaller than the Monte Carlo standard errors of approximately $0.033$~m/s reported in Table~\ref{tab:mc_stability}, confirming that the three models cannot be statistically distinguished on the basis of probabilistic accuracy alone.

All three SDE formulations are statistically indistinguishable in 
terms of probabilistic forecast accuracy: the mean CRPS differences 
of less than $0.006$\,m/s fall well below the Monte Carlo standard 
error of approximately $0.033$\,m/s, confirming that model selection 
cannot be justified on statistical grounds alone. As reported in Table~\ref{tab:comp_cost}, the diffusion-first model completes the full Monte Carlo simulation in approximately $30$~seconds, compared to $150$~seconds for the drift-based model and $210$~seconds for the Ornstein-Uhlenbeck formulation, representing approximately a factor-of-seven reduction in computational cost with no loss in statistical accuracy.

\begin{table}[h!]
\centering
\caption{Computational cost for generating 100 Monte Carlo simulations.}
\begin{tabular}{l c}
\hline
\textbf{Model} & \textbf{Simulation Time (s)} \\
\hline
Diffusion-First FP & 30 \\
Drift-Based FP     & 150 \\
Ornstein-Uhlenbeck (OU) & 210  \\
\hline
\end{tabular}
\label{tab:comp_cost}
\end{table}

\section{Wind Power Estimation and Analysis}
\label{sec:power}

We now use the conditional wind-speed ensembles obtained in Section~4 to estimate wind power.
The learned mapping is denoted

\begin{equation}
    P = f(V),
    \label{eq:power_map}
\end{equation}

where $f:\mathbb{R}\rightarrow\mathbb{R}$ is the nonlinear function approximated by the XGBoost model. Prior studies have demonstrated 
that XGBoost provides an accurate and flexible estimator of the nonlinear wind power curve, capturing physical features such as the cubic ramp-up region, rated power plateau, and cut-out behaviour without imposing a rigid parametric form~\cite{Xiong2022,Zheng2019}. The precise internal boosting mechanism is not central to the theoretical development here; what is important is that XGBoost yields a stable, piecewise-constant approximation of $f$  \cite{ChenG16} that generalises reliably to unseen wind speed inputs.

\begin{table}[ht]
\centering
\caption{XGBoost model performance metrics}
\label{tab:xgb_performance}

\scriptsize
\begin{tabular}{lcccc}
\hline
& RMSE (kW) & MAE (kW) & NRMSE (\%) & $R^2$ \\
\hline

Training & 73.59 & 46.89 & 3.59 & 0.9862 \\

Test     & 73.93 & 47.19 & 3.61 & 0.9862 \\

\hline
\end{tabular}
\normalsize

\end{table}

Table~\ref{tab:xgb_performance} summarises the performance of the XGBoost power curve model. The model demonstrates strong predictive accuracy, achieving an $R^2$ of 0.9862 and a normalised RMSE of 3.61\% of the rated power on the test set. The near-identical training and test metrics with an RMSE difference of only 0.34~kW and no change in $R^2$ to four decimal places—indicate good generalisation to unseen data without evidence of overfitting, hence supporting its use as the deterministic wind-to-power conversion stage within the proposed probabilistic simulation framework.
Once $f$ is established, the stochastic wind speed ensembles 
generated by each SDE model in the previous section are transformed into power predictions through this mapping.
Consequently, the uncertainty associated with wind speed $V$ is 
propagated into the power domain conditionally on the fitted function $f$. By the law of the unconscious statistician \cite{CasellaBerger2002}, the resulting power 
distribution inherits its statistical properties from the underlying distribution of the simulated wind speed, with uncertainty in $f$ itself ignored as a simplifying assumption.

\begin{equation}
    \mathcal{L}(P) = \mathcal{L}(f(V)),
    \label{eq:power_dist}
\end{equation}

where $\mathcal{L}(\cdot)$ denotes the probability law of a 
random variable. 

Because the central operational question is how reliably the 
turbine delivers at least a given power level, we report the 
complementary cumulative distribution function,

\begin{equation}
    F_{P}(p) = \mathbb{P}(P \le p),
    \label{eq:cdf}
\end{equation}

together with the corresponding exceedance probability,

\begin{equation}
    \mathbb{P}(P > p) = 1 - F_{P}(p),
    \label{eq:exceedance}
\end{equation}

which directly quantifies supply reliability at each threshold. 
For the Monte Carlo ensemble of $M$ pooled simulated power 
values $\{P_{j}^{\mathrm{sim}}\}_{j=1}^{M}$, the empirical exceedance probability is estimated as
\begin{equation}
    \widehat{\mathbb{P}}^{\,\mathrm{sim}}(P > p)
    \;=\;
    \frac{1}{M}\sum_{j=1}^{M}
    \mathbf{1}\!\left\{P_{j}^{\mathrm{sim}} > p\right\}.
    \label{eq:mc_exceedance}
\end{equation}

The full distribution of simulated power is the push-forward measure of the wind speed ensemble through the deterministic power-curve mapping $f_{\mathrm{XGB}}$,

\begin{equation}
    \mathcal{L}(\hat{P}) 
    \;=\; 
    \mathcal{L}\!\left(f_{\mathrm{XGB}}(\hat{V})\right),
    \label{eq:pushforward}
\end{equation}

so that every difference among the three power distributions is inherited entirely from the corresponding wind speed simulation 
differences. This statement holds conditionally on the fitted  function $f$: the reported predictive intervals ignore additional sources of uncertainty, including power-curve model error, measurement noise, curtailment events, pitch-control behaviour, 
and turbine state dependence. A more complete treatment would augment the deterministic map with a heteroskedastic residual model
\begin{equation}
P = f(V) + \varepsilon(V), \qquad 
\mathbb{E}[\varepsilon(V)\mid V] = 0, \qquad 
\text{Var}(\varepsilon(V)\mid V) = s^2(V)
\end{equation}
which would be particularly important near cut-in speed, near rated speed, and during pitch-control saturation. The deterministic map is retained here as a simplifying assumption.

\begin{figure}[h!]
\centering
\includegraphics[width=0.88\linewidth]{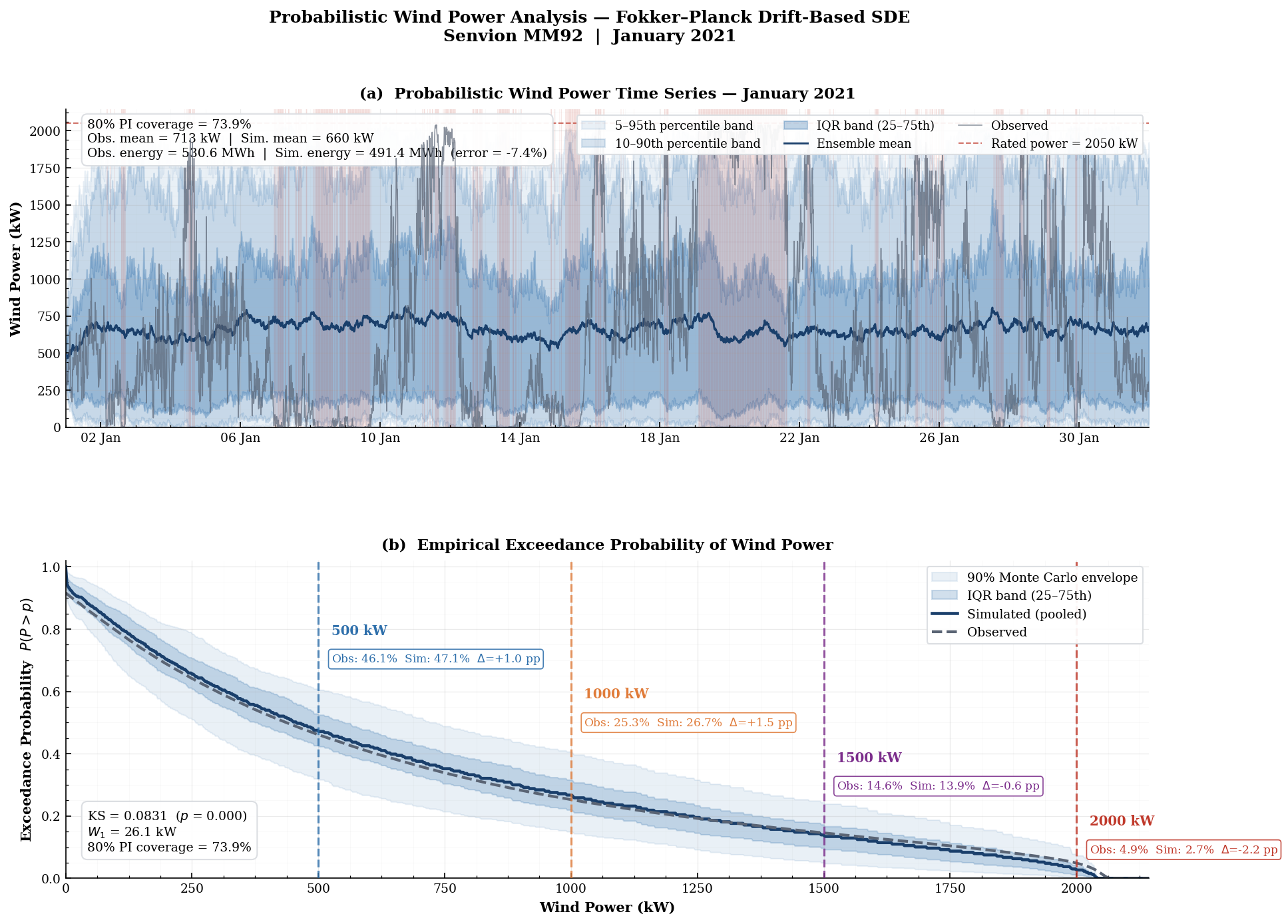}
\caption{Probabilistic wind power forecast, January~2021, Fokker-Planck drift-based SDE.
(a)~ensemble time series with percentile bands;
(b)~empirical exceedance probability.}
\label{fig:Power-Drif1}
\end{figure}

\begin{figure}[h!]
\centering
\includegraphics[width=0.88\linewidth]{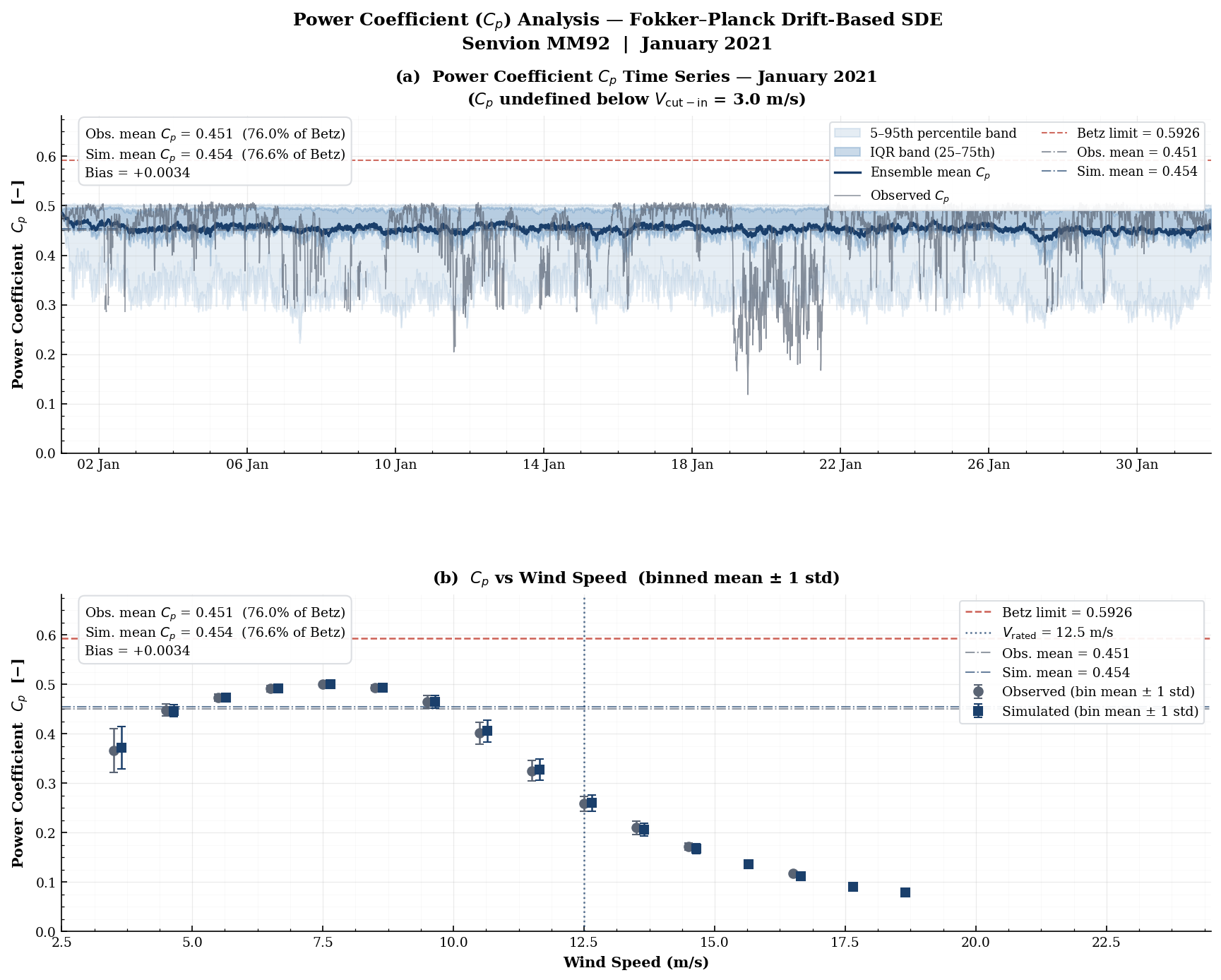}
\caption{Power coefficient $C_p$ analysis, January~2021,
Fokker-Planck drift-first SDE:
(a)~$C_p$ time series vs.\ observed;
(b)~binned $C_p$ vs.\ wind speed.
}
\label{fig:Power-Drift2}
\end{figure}

\begin{figure}[h!]
\centering
\includegraphics[width=0.88\linewidth]{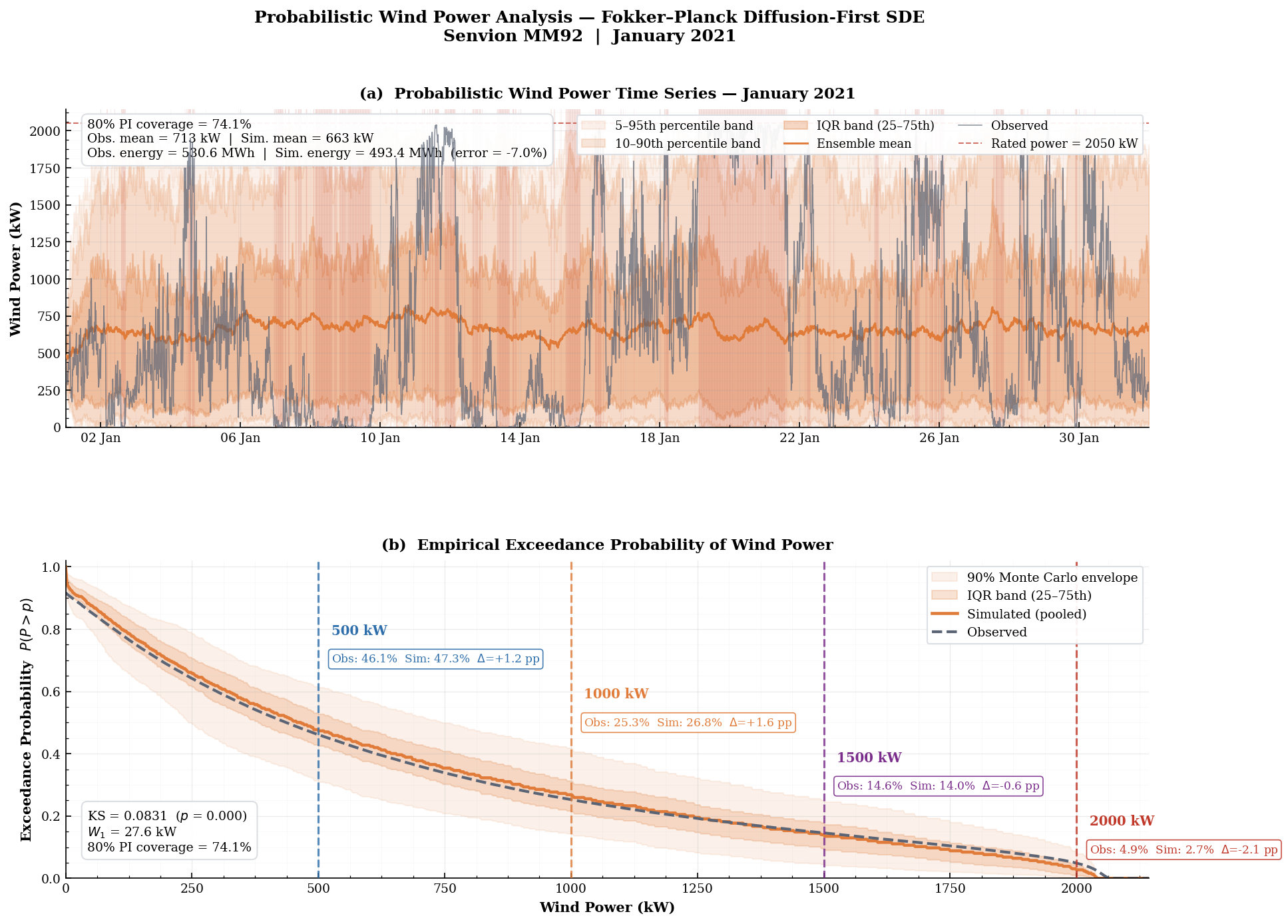}
\caption{Probabilistic wind power forecast, January~2021, Fokker-Planck diffusion-based SDE.
(a)~ensemble time series with percentile bands;
(b)~empirical exceedance probability}
\label{fig:Power-Diff1}
\end{figure}

\begin{figure}[h!]
\centering
\includegraphics[width=0.88\linewidth]{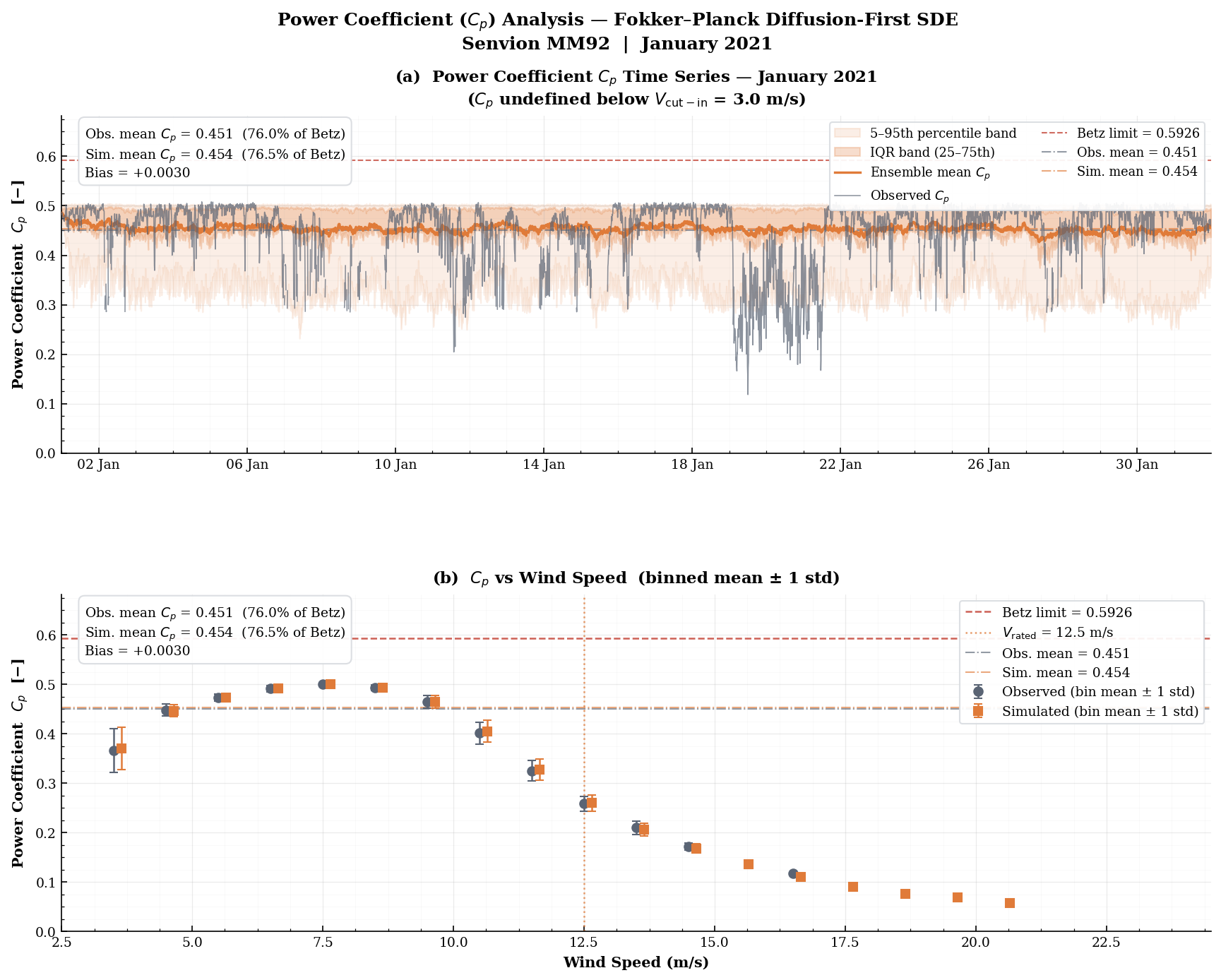}
\caption{Power coefficient $C_p$ analysis, January~2021,
Fokker-Planck diffusion-first SDE:
(a)~$C_p$ time series vs.\ observed;
(b)~binned $C_p$ vs.\ wind speed. }
\label{fig:Power-diff2}
\end{figure}

\begin{figure}[h!]
\centering
\includegraphics[width=0.88\linewidth]{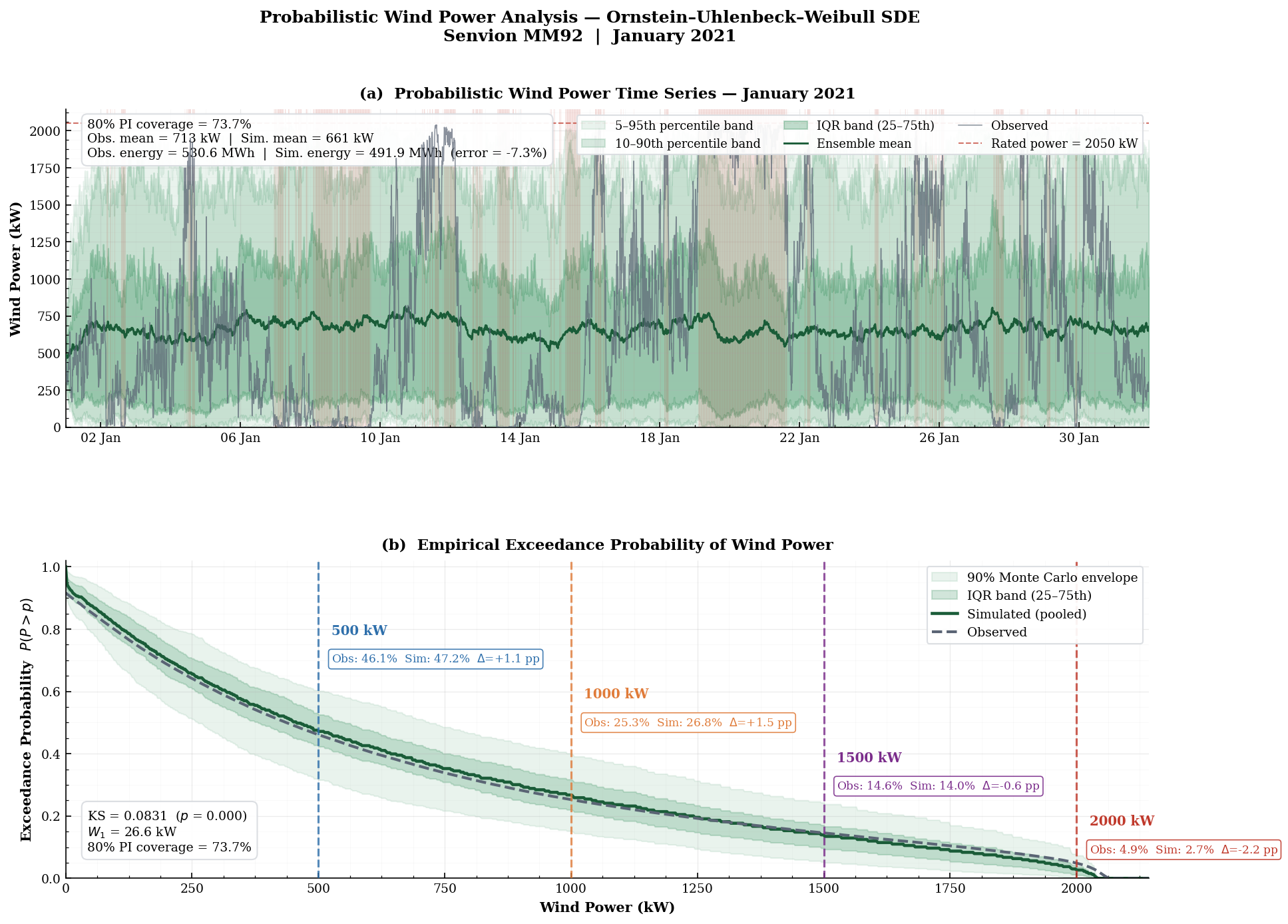}
\caption{Probabilistic wind power forecast, January~2021, Ornstein-Uhlenbeck-Weibull SDE. (a)~Ensemble time series with percentile bands; (b)~empirical exceedance probability.}
\label{fig:Power-OU1}
\end{figure}

\begin{figure}[h!]
\centering
\includegraphics[width=0.88\linewidth]{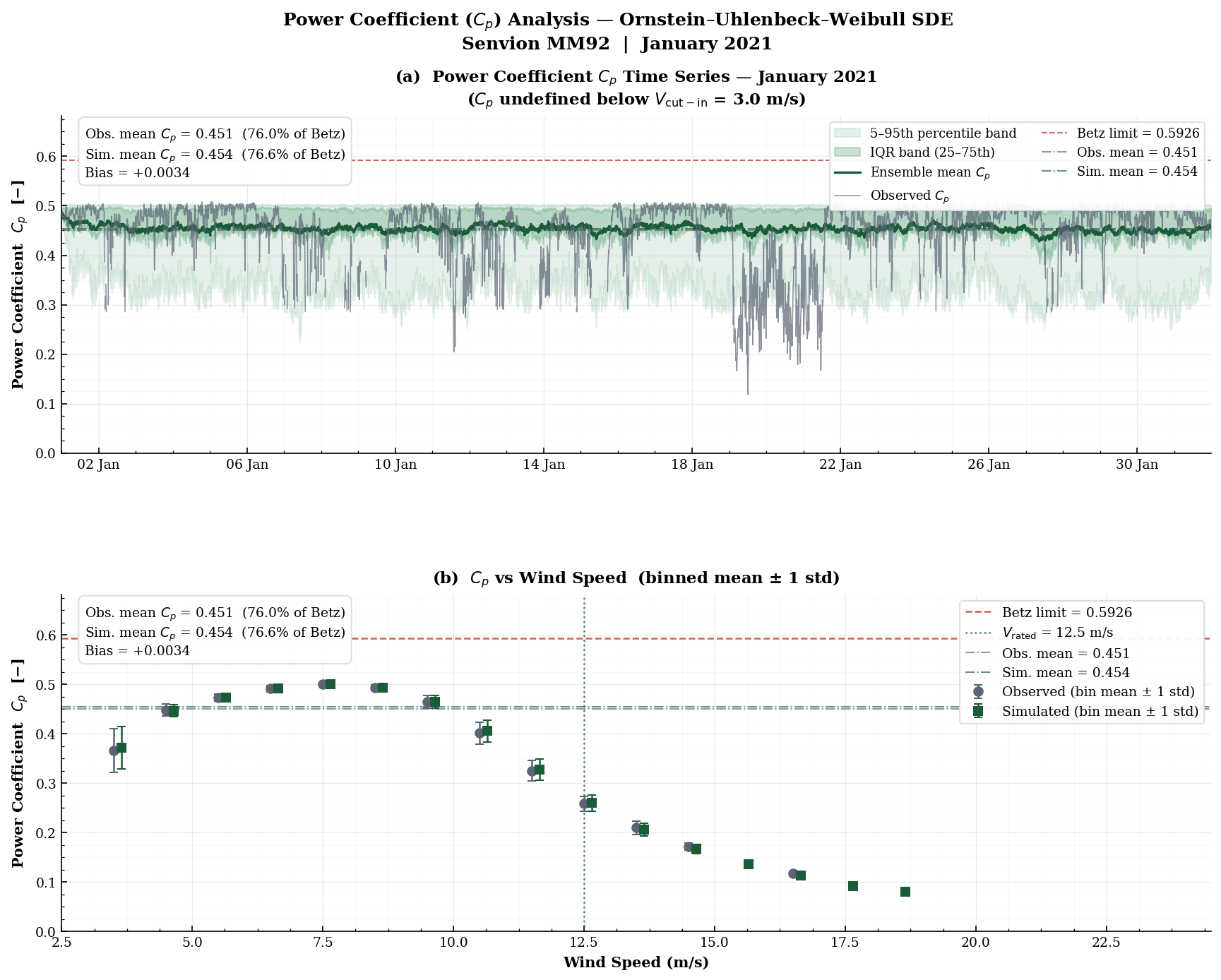}
\caption{Power coefficient $C_p$ analysis, January~2021,
Ornstein–Uhlenbeck-Weibull SDE:
(a)~$C_p$ time series vs.\ observed;
(b)~binned $C_p$ vs.\ wind speed. }
\label{fig:Power-OU2}
\end{figure}

\subsection{ Result Analysis of Wind Power}
Figures~\ref{fig:Power-Drif1}, \ref{fig:Power-Diff1},
and~\ref{fig:Power-OU1} present the probabilistic wind-power output for January 2021 obtained by mapping each of the $B = 100$ simulated wind-speed trajectories through the SCADA-calibrated power curve. The ensemble mean power lies in the narrow range $660$-$663$~kW, against an observed monthly mean of $713$~kW. The corresponding monthly energy yield is computed as
\begin{equation}\label{eq:energy_yield}
\hat{E}_{\rm month} = N \Delta t \, \bar{P}_{\rm sim},
\end{equation}
where $N = 4{,}464$ is the number of ten-minute intervals in January 2021 and $\Delta t = 1/6$~h. The relative energy bias is defined by
\begin{equation}\label{eq:bias}
\frac{\bar{E}_{\rm sim} - E_{\rm obs}}{E_{\rm obs}}.
\end{equation}
With $E_{\rm obs}=530.6$~MWh, the bias is $-7.4\%$ for the Fokker-Planck drift-based model $\bar{E}_{\rm sim}=491$~MWh, $-7.0\%$ for the Fokker-Planck diffusion-first model $\bar{E}_{\rm sim}=493$~MWh, and $-7.3\%$ for the Ornstein-Uhlenbeck-Weibull model $\bar{E}_{\rm sim}=492$~MWh. The near-identical bias across all three formulations confirms that the shortfall is not caused by the choice of SDE formulation. Rather, all three models share the same forecasted invariant law $(\hat{k}_{49}, \hat{\lambda}_{49})$ and the same decorrelation rate $\hat{\alpha}$; hence they all propagate conditional SDE path uncertainty around the same slightly underestimated wind-resource level.

\begin{table}[h!]
\centering
\caption{Observed and simulated wind power exceedance
and non-exceedance probabilities at four operational thresholds
($n = 4{,}464$ observations; $B = 100$ simulations).}
\label{tab:power_exceedance_full}
\small
\resizebox{\linewidth}{!}{
\begin{tabular}{lcccccccc}
\hline
 & \multicolumn{4}{c}{Exceedance $\mathbb{P}(P>p)$ (\%)}
 & \multicolumn{4}{c}{Non-exceedance $F_P(p)$ (\%)} \\
\textbf{Model}
  & $>$500\,kW & $>$1000\,kW & $>$1500\,kW & $>$2000\,kW
  & $\leq$500\,kW & $\leq$1000\,kW
  & $\leq$1500\,kW & $\leq$2000\,kW \\
\hline
Observed
  & 46.1 & 25.3 & 14.6 & 4.9
  & 53.9 & 74.7 & 85.4 & 95.1 \\
FP Drift-Based
  & 47.1 & 26.7 & 13.9 & 2.7
  & 52.9 & 73.3 & 86.1 & 97.3 \\
FP Diffusion-First
  & 47.3 & 26.8 & 14.0 & 2.7
  & 52.7 & 73.2 & 86.0 & 97.3 \\
OU-Weibull
  & 47.2 & 26.8 & 14.0 & 2.7
  & 52.8 & 73.2 & 86.0 & 97.3 \\
\hline
\end{tabular}
}
\end{table}

Table~\ref{tab:power_exceedance_full} summarises the exceedance
and non-exceedance probabilities at four operationally significant power thresholds for all three models alongside the observed values.
Turning to the exceedance probability curves in panel~(b) of Figures~\ref{fig:Power-Drif1}, \ref{fig:Power-Diff1}, and~\ref{fig:Power-OU1},  our simulations reproduce the observed wind power distribution with good fidelity across
the full operational range from 0 to 2050\,kW. The KS statistic is 0.0831 for all three formulations, and the Wasserstein distance $W_1$ lies between 26.1 and 27.6\,kW - less than 1.5\,\% of rated capacity. At the 500\,kW threshold, the simulated exceedance probability of approximately 47\,\% is within 1.2 percentage points of the observed 46.1\,\%, and at 1000\,kW the gap is 1.5 percentage points.
The only notable discrepancy occurs near rated power (2000-2050\,kW), where all three models underestimate the observed exceedance probability by 2.1-2.2 percentage points. The discrete probability spike at rated power in the observed power distribution is produced by the flat portion of the turbine power curve above rated speed, and its probability mass equals

\begin{equation}
\mathbb{P}(P = P_{\rm rated}) = \mathbb{P}(V \geq V_{\rm rated})
\end{equation}
when the map is exactly saturated. The observed underestimation near $2000$ kW therefore reflects two compounding factors: first, the simulated wind-speed distribution assigns insufficient probability to the set $\{v : v \geq V_{\rm rated}\}$, as evidenced by the shared upper-tail compression across all three SDE formulations; and second, the learned XGBoost approximation may introduce smoothing or bias in the plateau region that prevents exact reproduction of this discrete probability concentration at rated power\cite{loukatou2018stochastic}. 
Regarding prediction interval coverage, all three models achieve coverage in the range 73.7-74.1\,\% against a nominal 80\,\%. Prediction interval coverage is a repeated-sampling property that 
is properly assessed over many independent forecast cycles rather than a single realisation. The observed coverage rates of $73.7\%$, 
$74.1\%$, and $73.7\%$, therefore carry limited diagnostic value in isolation. The concentration of exceedances during the $18$-$22$ January storm period is consistent with a known structural property of stationary Weibull models: such models assign low probability 
mass to persistent high-wind regimes, making sustained storm episodes structurally difficult to envelope regardless of ensemble size or SDE formulation. The fact that all three structurally 
different SDE formulations yield identical sub-nominal coverage confirms that this reflects a shared property of the Weibull invariant law rather than a deficiency of any individual model.

\paragraph{Power Coefficient Validation}

To assess the physical consistency of the composed map 
$V \mapsto f(V) \mapsto C_p$, we compute the power coefficient 
for both the observed and simulated series. This diagnostic does not constitute an independent physical validation of the SDE trajectories themselves, but rather confirms that the simulated wind speeds, once passed through the deterministic power curve, remain compatible with the expected aerodynamic operating envelope of the Senvion MM92 rotor.  The power coefficient is defined as

\begin{equation}
    C_p = \frac{P}{\dfrac{1}{2}\,\rho\,A\,V^3},
    \label{eq:cp}
\end{equation}

where $P$ is the turbine power output in watts, 
$\rho = 1.225$\,kg\,m$^{-3}$ is the air density at standard 
conditions, $A = \pi R^2$ is the rotor swept area with blade 
radius $R = 46$\,m, and $V$ is the hub-height wind 
speed~\cite{Burton2011}. The power coefficient 
$C_p$ quantifies the fraction of the kinetic energy flux through 
the rotor disk that is converted into electrical power, by the 
Betz limit, it is theoretically bounded above by 
$C_p^{\max} = 16/27 \approx 0.593$~\cite{Burton2011}. The results are presented in two complementary 
forms for each model: a time series of $C_p$ throughout 
January~2021 and a binned $C_p$ versus wind speed profile, 
shown for the drift-based, diffusion-first, and OU-Weibull 
formulations in Figures~\ref{fig:Power-Drift2}, 
\ref{fig:Power-diff2}, and~\ref{fig:Power-OU2} respectively. 
This diagnostic provides a direct physical check on whether our simulated wind speeds, once passed through the power curve, remain compatible with the expected aerodynamic operating envelope, particularly in the below-rated speed region. 
Since $C_p$ is computed from the same simulated wind-speed 
values that enter the power-curve mapping $P = f(V)$, the 
diagnostic primarily validates the consistency of the 
composed map $V \to f(V) \to C_p$.

All three SDE formulations reproduce the observed mean $C_p$ of 0.451 (76.0\,\% of the Betz limit) with a simulated mean of
0.454 and a bias of only 0.003, which is negligible from an operational standpoint. The ensemble bands in panel~(a) span
roughly 0.30 to 0.55 throughout the month, consistent with the observed $C_p$ variability, and the ensemble mean tracks the
observed series without any systematic drift across the full 31-day period. The sharp drop in observed $C_p$ to approximately 0.15 around 20-22 January, which corresponds to the above-rated storm episode identified in the wind speed analysis, is correctly captured within the ensemble envelope of all three models.

The binned profiles in panel~(b) provide the most physically informative validation. Across the entire below-rated operating envelope from $V_{\mathrm{cut-in}} = 3.0$\,m\,s$^{-1}$ to $V_{\mathrm{rated}} = 12.5$\,m\,s$^{-1}$, our simulated bin means lie within one standard deviation of the observed values
in every wind speed bin. The $C_p$ curve peaks
near 0.50 in the 7.5-10\,m\,s$^{-1}$ range, which corresponds to the design operating point of the Senvion MM92 rotor, and
this peak is correctly reproduced by all three formulations. Above $V_{\mathrm{rated}}$, both the observed and simulated
$C_p$ values decline steeply, as expected from blade pitch regulation \cite{Burton2011}: the turbine deliberately reduces aerodynamic
efficiency to hold output at the rated 2050\,kW. 
Taken together, the $C_p$ analysis confirms that the three SDE formulations produce wind-speed trajectories whose simulated power distribution, after deterministic power-curve mapping, remains compatible with the expected aerodynamic operating envelope of the Senvion MM92 rotor, particularly in the below-rated speed region. This diagnostic primarily validates the consistency of the composed map $V \mapsto f(V) \mapsto C_p$.

\subsection{Practical Implications for Power-System Operation}

The exceedance and non-exceedance probabilities in Table~8 should be interpreted as decision-relevant probabilistic summaries rather than as completed operational decisions. They provide the marginal distributional inputs needed by reserve, storage, market, and maintenance models, but they do not solve the corresponding downstream optimization problems. For instance, the fact that the turbine is below $500$~kW in about $53.9\%$ of the ten-minute intervals is useful for reserve-margin assessment, but the actual reserve decision also depends on contractual obligations, imbalance costs, response times, and the temporal clustering of deficits. Similarly, storage sizing cannot be inferred from a non-exceedance probability alone. For a contractual delivery floor $c$, the relevant mathematical quantities involve path functionals such as
\begin{equation}
\int_0^T (c-P_t)^+\,\dd t,
\end{equation}
or the distribution of maximal cumulative deficit over delivery windows. The conditional ensemble generated in this paper supplies samples from which such quantities can be estimated, but the optimization problem itself is outside the present scope.

The same distinction applies to market bidding \cite{Benth2018} and maintenance planning. Lower percentiles of the monthly energy distribution can inform conservative bid volumes. In contrast, upper percentiles describe upside potential, but a bidding strategy requires an explicit market loss function and penalty structure. Likewise, high-power exceedance frequencies provide information about the loading regime, but fatigue-informed maintenance requires a load model or damage functional rather than threshold probabilities alone. The contribution of the present paper is therefore to construct and validate the conditional probabilistic input distribution on which these operational decisions may be built.

\section{Conclusion}
\label{sec:conclusion}

This work develops a conditional one-month-ahead probabilistic framework for wind-power simulation at ten-minute resolution, applied to Turbine~1 of the Kelmarsh Wind Farm for January 2021. The monthly Weibull parameters are estimated with serial-dependence-corrected covariance and forecast by a heteroskedastic Kalman filter on a bivariate VAR(1) state-space model. The resulting predictive law is summarized, for the numerical SDE study, by the MMSE point forecast $(\hat{k}_{49},\hat{\lambda}_{49})=(2.6272,7.0691\,\mathrm{m\,s}^{-1})$, which defines the invariant Weibull law imposed on all simulated wind-speed processes.

Conditional on this common invariant law, three Weibull-stationary SDE constructions are compared: the OU-Weibull transport model, the drift-first Fokker-Planck model, and the diffusion-first Fokker-Planck model. Their mean CRPS values, ranging from $1.569$ to $1.575\,\mathrm{m\,s}^{-1}$, differ by less than the Monte Carlo standard error. Thus, the three models cannot be ranked statistically in this experiment. The diffusion-first model is preferred only because it achieves comparable probabilistic accuracy with substantially lower computational cost, completing the Monte Carlo simulation in about $30$ seconds compared with $150$ seconds for the drift-first model and $210$ seconds for the OU-Weibull model.

After deterministic XGBoost power-curve mapping, the simulated power distributions attain Wasserstein distances of $26.1$-$27.6$~kW from the observed January 2021 distribution, below $1.4\%$ of rated capacity. The simulated monthly energy yield is $491$-$493$~MWh against an observed value of $530.6$~MWh, giving a bias of about $-7.3\%$. This bias is common to all three SDE constructions, and it is best interpreted as a limitation of the one-month-ahead forecasted invariant law for a month containing a sustained storm episode, not as evidence against one particular SDE formulation. The exceedance analysis and power-coefficient diagnostics show that the framework gives useful conditional probabilistic inputs for power-system studies. Still, the paper does not claim to solve the downstream reserve, storage, bidding, or fatigue-optimization problems.

The present contribution should therefore be read as a rigorously structured conditional Weibull-stationary SDE forecasting framework. Its natural next step is to replace the MMSE plug-in law by the full mixture over the Kalman predictive distribution of $(k,\lambda)$ and to test alternative invariant families for months or sites whose wind regimes are not adequately captured by a single Weibull law.

\appendix

\section{Mathematical Proofs}
\label{app:proofs}

\subsection{Global Positivity of the Drift-First Diffusion}
\label{app:lemma}

\begin{lemma}\label{lem:positivity}
Let $U$ have Weibull density $p_W$ on $(0,\infty)$ and finite mean $\mu_W$. For every finite $x>0$,
\begin{equation}
\int_0^x(\mu_W-u)p_W(u)\,\dd u
=F_W(x)\bigl(\mu_W-\E[U\mid U\le x]\bigr)>0.
\end{equation}
Consequently, the drift-first diffusion coefficient
\begin{equation}
 b^2(x)=\frac{2\alpha}{p_W(x)}\int_0^x(\mu_W-u)p_W(u)\,\dd u
\end{equation}
is strictly positive for every finite $x>0$.
Moreover,
\begin{equation}
 b^2(x)\sim \frac{2\alpha\lambda^k}{k}x^{2-k},\qquad x\to\infty,
\end{equation}
so upper-tail attenuation holds only when $k>2$.
\end{lemma}

\begin{proof}
The identity follows by writing the truncated expectation explicitly. To prove strict positivity, decompose the mean as
\begin{equation}
\E[U]=F_W(x)\E[U\mid U\le x]+\bigl(1-F_W(x)\bigr)\E[U\mid U>x].
\end{equation}
Since the Weibull law is non-degenerate and has positive mass on both sides of every finite threshold, the conditional mean above $x$ is strictly larger than the conditional mean below $x$. Hence $\E[U\mid U\le x]<\E[U]=\mu_W$, and the numerator in $b^2(x)$ is strictly positive. The asymptotic formula is obtained from the upper incomplete gamma expansion applied to
\begin{equation}
\int_0^x(\mu_W-u)p_W(u)\,\dd u=\int_x^\infty(u-\mu_W)p_W(u)\,\dd u.
\end{equation}
Dividing the resulting expansion by $p_W(x)$ gives $b^2(x)\sim 2\alpha\lambda^k x^{2-k}/k$.

\end{proof}

\subsection{Pearson Autocorrelation of the OU-W Process}
\label{app:pearson}

Let $h_n(x) = H_n(x)/\sqrt{n!}$ be the orthonormal Hermite functions
in $L^2(\phi)$. Since $g \in L^2(\phi)$, there exists an expansion
$g(x) = \sum_{n=0}^{\infty} a_n h_n(x)$ with
$a_n = \mathbb{E}[g(Z)h_n(Z)]$, $Z \sim \mathcal{N}(0,1)$.
For $(X,Y)$ bivariate normal with correlation $\rho$, the Mehler
identity gives $\mathbb{E}[g(X)g(Y)] = \sum_{n=0}^{\infty} a_n^2\,\rho^n$.
Since $\mathbb{E}[g(X)] = a_0$ and
$\mathrm{Var}(g(X)) = \sum_{n=1}^{\infty} a_n^2$, the Pearson
autocorrelation of $V_t = g(X_t)$ at lag $\tau$ is
\begin{equation}\label{eq:pearson_ou}
    \mathrm{Corr}(V_t, V_{t+\tau})
    = \frac{\sum_{n=1}^{\infty} a_n^2\,e^{-n\alpha\tau}}
           {\sum_{n=1}^{\infty} a_n^2}.
\end{equation}
This is a mixture of exponentials rather than a single exponential.
It reduces to $e^{-\alpha\tau}$ only if $a_n = 0$ for all $n \geq 2$,
which requires $g$ to be affine - incompatible with mapping a Gaussian
distribution to a Weibull distribution. Pearson correlation must
therefore be treated as a derived quantity computed numerically
from equation~\eqref{eq:pearson_ou}.

\vspace{1.2 em}

\noindent\textbf{CRediT authorship contribution statement}

\vspace{0.8 em}

\textbf{Luca Di Persio:} Conceptualization, Methodology, Software, Validation, Formal analysis, Investigation, Data curation, Writing – original draft, Writing – review and editing, Visualization,  Project administration.

\vspace{0.4 em}

\textbf{Mehrdad Ghadiri:} Conceptualization, Methodology, Software, Validation, Formal analysis, Investigation, Data curation, Writing – original draft, Writing – review and editing, Visualization, Project administration. 

\vspace{1.2 em}

\noindent\textbf{Declaration of competing interest}

\vspace{0.8 em}

The authors declare that they have no known competing financial
interests or personal relationships that could have appeared to influence
the work reported in this paper.

\vspace{1 em}

\noindent\textbf{Data availability}

\vspace{0.8 em}

The SCADA data of the examined wind turbine is available in the Zenodo repository \cite{zenodo_dataset}.


\begin{thebibliography}{00}


\bibitem{JakhmolaWind2026}
Jakhmola, A., Jewell, J., Vinichenko, V. et al. Probabilistic projections of global wind and solar power growth based on historical national experience. Nat Energy (2026).
DOI: \url{https://doi.org/10.1038/s41560-026-02021-w}.

\bibitem{SelvarajWind2026}
Selvaraj, J., Muthuramalingam, L., Karthikeyan, V. et al. Optimizing Wind Energy Integration: A Review of Forecasting Techniques and Emerging Trends. Arch Computat Methods Eng 33, 4261–4286 (2026).
DOI: \url{https://doi.org/10.1007/s11831-025-10442-1}.

\bibitem{AI2023114222}
Ai, C., He, S., Hu, H., Fan, X. and Wang, W. Chaotic time series wind power interval prediction based on quadratic decomposition and intelligent optimization algorithm. Chaos, Solitons \& Fractals, 177, 114222 (2023). DOI: \url{https://doi.org/10.1016/j.chaos.2023.114222}.

\bibitem{ZHU2024118062}
Zhu, J., He, Y., Yang, X. and Yang, S. Ultra-short-term wind power probabilistic forecasting based on an evolutionary non-crossing multi-output quantile regression deep neural network. Energy Conversion and Management, 301, 118062 (2024). DOI: \url{https://doi.org/10.1016/j.enconman.2024.118062}.

\bibitem{XIE2023105804}
Xie, Y., Li, C., Li, M., Liu, F. and Taukenova, M. An overview of deterministic and probabilistic forecasting methods of wind energy. iScience, 26(1), 105804 (2023). DOI: \url{https://doi.org/10.1016/j.isci.2022.105804}.

\bibitem{SAEED2025137979}
Saeed, et al. Enhanced wind speed forecasting for sustainable power systems: A deep learning framework unifying deterministic predictions and uncertainty quantification. Energy, 335, 137979 (2025). DOI: \url{https://doi.org/10.1016/j.energy.2025.137979}.

\bibitem{DantasBrowell2026}
Dantas, G. and Browell, J. Seamless Short- to Mid-Term Probabilistic Wind Power Forecasting. Wind Energy, 29(2), e70079 (2026). DOI: \url{https://doi.org/10.1002/we.70079}.

\bibitem{ARSLANTUNCAR2024197}
Arslan Tuncar, E., Sağlam, Ş. and Oral, B. A review of short-term wind power generation forecasting methods in recent technological trends. Energy Reports, 12, 197-209 (2024). DOI: \url{https://doi.org/10.1016/j.egyr.2024.06.006}.

\bibitem{GUO20253753}
Guo, X., Zeng, P., Xiong, X., Wang, G. and Cui, Y. Short-term wind power forecasting methods based on machine learning: A review and case study. Energy Reports, 14, 3753-3782 (2025). DOI: \url{https://doi.org/10.1016/j.egyr.2025.10.040}.

\bibitem{electricity6030048}
Ghadiri, M. and Di Persio, L. Hybrid SDE-Neural Networks for Interpretable Wind Power Prediction Using SCADA Data. Electricity, 6(3), 48 (2025). URL: \url{https://www.mdpi.com/2673-4826/6/3/48}.


\bibitem{DiPersio2024WindItaly}
Di Persio, L., Fraccarolo, N. and Veronese, A. Wind Energy Production in Italy: A Forecasting Approach Based on Fractional Brownian Motion and Generative Adversarial Networks. Mathematics, 12(13) (2024). DOI: \url{https://doi.org/10.3390/math12132105}.

\bibitem{Ceresa2024}
Ceresa, G., Trevisiol, A., Rapizza, M.R. and Cirio, D. Stochastic Simulation of Wind Power Profiles from Time Series Analysis Considering Dependencies on Meteorological Variables. In: Pong, P. (eds) Renewable Energy Resources and Conservation. Green Energy and Technology. Springer, Cham (2024). DOI: \url{https://doi.org/10.1007/978-3-031-59005-4_10}.


\bibitem{Iversen2017}
Iversen, E.B., Morales, J.M., Møller, J.K., Trombe, P.-J. and Madsen, H. Leveraging stochastic differential equations for probabilistic forecasting of wind power using a dynamic power curve. Wind Energy, 20, 33-44 (2017). DOI: \url{https://doi.org/10.1002/we.1988}.

\bibitem{11088843}
Wu, Y., Du, W., Wu, W., Yin, T., Wang, Y. and Li, P. Medium and Long-Term Power Generation Forecasting and Operation Strategies for Wind Farms Based on Deep Learning. In: 2025 IEEE 2nd International Conference on Deep Learning and Computer Vision (DLCV), pp. 1-6 (2025). DOI: \url{https://doi.org/10.1109/DLCV65218.2025.11088843}.


\bibitem{ladopoulou2025probabilisticwindpowermodelling}
Ladopoulou, D., Hong, D.M. and Dellaportas, P. Probabilistic Wind Power Modelling via Heteroscedastic Non-Stationary Gaussian Processes. arXiv preprint arXiv:2505.09026 (2025). Available at: \url{https://arxiv.org/abs/2505.09026}.

\bibitem{ZARATEMINANO201342}
Zárate-Miñano, R., Anghel, M. and Milano, F. Continuous wind speed models based on stochastic differential equations. Applied Energy, 104, 42-49 (2013). DOI: \url{https://doi.org/10.1016/j.apenergy.2012.10.064}.

\bibitem{ZARATEMINANO2016186}
Zárate-Miñano, R. and Milano, F. Construction of SDE-based wind speed models with exponentially decaying autocorrelation. Renewable Energy, 94, 186-196 (2016). DOI: \url{https://doi.org/10.1016/j.renene.2016.03.026}.

\bibitem{ARENASLOPEZ2020113152}
Arenas-López, J.P. and Badaoui, M. A Fokker–Planck equation based approach for modelling wind speed and its power output. Energy Conversion and Management, 222, 113152 (2020). DOI: \url{https://doi.org/10.1016/j.enconman.2020.113152}.

\bibitem{ARENASLOPEZ2020118842}
Arenas-López, J.P. and Badaoui, M. The Ornstein-Uhlenbeck process for estimating wind power under a memoryless transformation. Energy, 213, 118842 (2020). DOI: \url{https://doi.org/10.1016/j.energy.2020.118842}.

\bibitem{zenodo_dataset}
Peder Bacher,
\textit{Wind Turbine SCADA Datasets},
2022.
Available at: \url{https://zenodo.org/records/5841834}.
Accessed August 28, 2024.


\bibitem{MORRISON2022473}
Morrison, R., Liu, X. and Lin, Z. Anomaly detection in wind turbine SCADA data for power curve cleaning. Renewable Energy, 184, 473-486 (2022). DOI: \url{https://doi.org/10.1016/j.renene.2021.11.118}.


\bibitem{Gill2012}
Gill, S., Stephen, B. and Galloway, S. Wind Turbine Condition Assessment Through Power Curve Copula Modeling. IEEE Transactions on Sustainable Energy, 3(1), 94-101 (2012). DOI: \url{https://doi.org/10.1109/TSTE.2011.2167164}.


\bibitem{LONG2022118594}
Long, H., Xu, S. and Gu, W. An abnormal wind turbine data cleaning algorithm based on color space conversion and image feature detection. Applied Energy, 311, 118594 (2022). DOI: \url{https://doi.org/10.1016/j.apenergy.2022.118594}.


\bibitem{Burton2011}
Burton, T., Jenkins, N., Sharpe, D. and Bossanyi, E. \textit{Wind Energy Handbook}. 2nd ed. Wiley (2011). DOI: \url{https://doi.org/10.1002/9781119992714}.


\bibitem{KelmarshWindFarm}
Google Maps. \textit{Kelmarsh Wind Farm, Northamptonshire, United Kingdom}. Available at: \url{https://maps.google.com}. Accessed: 20 May 2026.

 
\bibitem{windturbinemodels_senvionmm92}
Senvion MM92 Wind Turbine Specifications,
\textit{Wind Turbine Models},
available at \url{https://en.wind-turbine-models.com/turbines/889-senvion-mm92},
accessed on March 22, 2026.


\bibitem{seguro2000modern}
J.V. Seguro and T.W. Lambert,
\textit{Modern estimation of Weibull parameters for wind energy applications},
Wind Engineering,
24(4),
247-262,
2000.
DOI: \url{https://doi.org/10.1260/0309524001495620}.


\bibitem{carta2009review}
J.A. Carta, P. Ramírez, and S. Velázquez,
\textit{A review of wind speed probability distributions used in wind energy analysis},
Renewable and Sustainable Energy Reviews,
13(5),
933-955,
2009.
DOI: \url{https://doi.org/10.1016/j.rser.2008.02.005}.


\bibitem{ALJEDDANI202356}
Aljeddani, S.M.A. and Mohammed, M.A. A novel approach to Weibull distribution for the assessment of wind energy speed. Alexandria Engineering Journal, 78, 56-64 (2023). DOI: \url{https://doi.org/10.1016/j.aej.2023.07.027}.


\bibitem{RAMIREZ20052419}
Ramírez, P. and Carta, J.A. Influence of the data sampling interval in the estimation of the parameters of the Weibull wind speed probability density distribution: a case study. Energy Conversion and Management, 46(15), 2419-2438 (2005). DOI: \url{https://doi.org/10.1016/j.enconman.2004.11.004}.


\bibitem{White1982}
White, H. Maximum Likelihood Estimation of Misspecified Models. Econometrica, 50(1), 1-25 (1982). DOI: \url{https://doi.org/10.2307/1912526}.


\bibitem{Andrews1991}
Andrews, D.W.K. Heteroskedasticity and Autocorrelation Consistent Covariance Matrix Estimation. Econometrica, 59(3), 817-858 (1991). DOI: \url{https://doi.org/10.2307/2938229}.


\bibitem{Lutkepohl2005}
Lütkepohl, H. \textit{New Introduction to Multiple Time Series Analysis}. Springer, Berlin Heidelberg (2005). DOI: \url{https://doi.org/10.1007/978-3-540-27752-1}.


\bibitem{lehmann1998theory}
Lehmann, E.L. and Casella, G. \textit{Theory of Point Estimation}. 2nd ed. Springer New York (1998). ISBN: 978-0-387-98502-2.

\bibitem{Kalman1960}
Kalman, R.E. A New Approach to Linear Filtering and Prediction Problems. Journal of Basic Engineering, 82(1), 35-45 (1960).


\bibitem{CARTA2009933}
Carta, J.A., Ramírez, P. and Velázquez, S. A review of wind speed probability distributions used in wind energy analysis: Case studies in the Canary Islands. Renewable and Sustainable Energy Reviews, 13(5), 933-955 (2009). DOI: \url{https://doi.org/10.1016/j.rser.2008.05.005}.


\bibitem{Haessig2015}
Haessig, P., Multon, B., Ahmed, H.B., Lascaud, S. and Bondon, P. Energy storage sizing for wind power: impact of the autocorrelation of day-ahead forecast errors. Wind Energy, 18, 43-57 (2015). DOI: \url{https://doi.org/10.1002/we.1680}.


\bibitem{Risken1996}
Risken, H. \textit{The Fokker-Planck Equation: Methods of Solution and Applications}. 2nd ed. Springer, Berlin (1996). ISBN: 978-3-540-61530-9.


\bibitem{Xiong2022}
Xiong, X., Guo, X., Zeng, P., Zou, R. and Wang, X. A Short-Term Wind Power Forecast Method via XGBoost Hyper-Parameters Optimization. Frontiers in Energy Research, 10, 905155 (2022). DOI: \url{https://doi.org/10.3389/fenrg.2022.905155}.


\bibitem{Zheng2019}
Zheng, H. and Wu, Y. A XGBoost Model with Weather Similarity Analysis and Feature Engineering for Short-Term Wind Power Forecasting. Applied Sciences, 9, 3019 (2019). DOI: \url{https://doi.org/10.3390/app9153019}.


\bibitem{ChenG16}
Chen, T. and Guestrin, C. XGBoost: A Scalable Tree Boosting System. CoRR, abs/1603.02754 (2016). URL: \url{http://arxiv.org/abs/1603.02754}.


\bibitem{CasellaBerger2002}
Casella, G. and Berger, R.L. \textit{Statistical Inference}. 2nd ed. Duxbury Press (2002). ISBN: 0-534-24312-6.


\bibitem{loukatou2018stochastic}
Loukatou, A., Howell, S., Johnson, P. and Duck, P. \textit{Stochastic wind speed modelling for estimation of expected wind power output}. Applied Energy, 228, 1328-1340 (2018). DOI: \url{https://doi.org/10.1016/j.apenergy.2018.06.117}.

%
\bibitem{Benth2018}
Benth, F.E., Di Persio, L. and Lavagnini, S. Stochastic modeling of wind derivatives in energy markets. Risks, 6(2) (2018). DOI: \url{https://doi.org/10.3390/risks6020056}.


\end{thebibliography}
\end{document}